\documentclass[lettersize,journal]{IEEEtran}
\usepackage{cite}
\usepackage{amsfonts,amssymb,amsthm,amsbsy,amsmath,paralist,bm,ifthen,color}
\usepackage{algorithm,algorithmic}
\usepackage{graphicx}
\usepackage{subfig}
\usepackage[font=footnotesize]{caption}
\usepackage{xcolor}
\usepackage{relsize}
\usepackage[numbers,sort&compress,square]{natbib}
\usepackage{tabularx}
\usepackage{array,booktabs}
\usepackage{multirow}
\usepackage{setspace,cases}
\usepackage{makecell}
\usepackage{cuted,lipsum}
\newtheorem{lemma}{Lemma}

\newtheorem{proposition}{Proposition}

\ifCLASSINFOpdf

\else

\fi

\begin{document}
\title{Robust Beamforming Design for IRS-Aided URLLC in D2D Networks}
\author{Jing Cheng,~\IEEEmembership{Graduate Student Member,~IEEE}, Chao Shen,~\IEEEmembership{Member,~IEEE}, Zheng Chen,~\IEEEmembership{Member,~IEEE}, Nikolaos Pappas,~\IEEEmembership{Senior Member,~IEEE}
\thanks{Manuscript received January 18, 2022; revised May 6, 2022 and June 4, 2022; accepted July 5, 2022. Date of publication xxx xx, 2022; date of current version xxx xx, 2022. This work was supported in part by the National Key R\&D Program of China under Grant 2021YFB2900301; in part by the State Key Laboratory of Rail Traffic Control and Safety, Beijing Jiaotong University under Contract RCS2021ZP002; in part by the NSFC, China under Grants 61871027, 62031008, and U1834210; and in part by the China Scholarship Council (CSC) under Grant 202007090174. The work of N. Pappas was supported by the Swedish Research Council (VR), ELLIIT, and CENIIT. A part of this work has been presented at the 25th International ITG Workshop on Smart Antennas (WSA), 2021 \cite{Cheng2021}. The editor coordinating the review of this article and approving it for publication was Behrooz Makki. (Corresponding author: Chao Shen.)}
\thanks{Jing Cheng is with the State Key Laboratory of Rail Traffic Control and Safety, Beijing Jiaotong University, Beijing 100044, China (e-mail: chengjing@bjtu.edu.cn).}
\thanks{Chao Shen is with the State Key Laboratory of Rail Traffic Control and Safety, Beijing Jiaotong University, Beijing 100044, China, and also with the Shenzhen Research Institute of Big Data, Shenzhen, China (email: chaoshen@sribd.cn).}
\thanks{Zheng Chen is with the Department of Electrical Engineering, Link\"oping University,  58183 Link\"oping, Sweden (e-mail: zheng.chen@liu.se).}
\thanks{Nikolaos Pappas is with the Department of Computer and Information Science, Link\"oping University, 58183 Link\"oping, Sweden (e-mail: nikolaos.pappas@liu.se).}
}

%
\maketitle
\vspace{-3mm}
\begin{abstract}
Intelligent reflecting surface (IRS) and device-to-device (D2D) communication are two promising technologies for improving transmission reliability between transceivers in communication systems. In this paper, we consider the design of reliable communication between the access point (AP) and actuators for a downlink multiuser multiple-input single-output (MISO) system in the industrial IoT (IIoT) scenario. We propose a two-stage protocol combining IRS with D2D communication so that all actuators can successfully receive the message from AP within a given delay. The superiority of the protocol is that the communication reliability between AP and actuators is doubly augmented by the IRS-aided first-stage transmission and the second-stage D2D transmission. A joint optimization problem of active and passive beamforming is formulated, which aims to maximize the number of actuators with successful decoding. We study the joint beamforming problem for cases where the channel state information (CSI) is perfect and imperfect. For each case, we develop efficient algorithms that include convergence and complexity analysis. Simulation results demonstrate the necessity and role of IRS with a well-optimized reflection matrix, and the D2D network in promoting reliable communication. Moreover, the proposed protocol can enable reliable communication even in the presence of stringent latency requirements and CSI estimation errors.
\end{abstract}

\IEEEpeerreviewmaketitle
\vspace{-3mm}
\section{Introduction}
With the development of the Internet of things (IoT) and the fifth-generation (5G) and beyond wireless networks, the communication paradigm shifts from human-to-human (H2H) communication to machine-to-machine (M2M) communication. One of the IoT use cases is critical industrial IoT (IIoT), which is envisioned to support mission-critical applications such as intelligent transportation systems, remote healthcare and smart manufacturing \cite{SuttonICST2019}. The implementation of critical IIoT requires establishing ultra-reliable and low-latency communication (URLLC) among IIoT devices. To meet the demanding low-latency with the order of milliseconds required for critical IIoT applications, short packets are usually transmitted. However, this will inevitably cause a loss of coding gain. That is, low latency is achieved at the expense of reliability. To enable ultra-reliable communication between the transmitter and the receiver, a retransmission mechanism \cite{KotabaIToWC2021} utilizing temporal diversity is proposed. However, URLLC packets scheduled in mini-slots are subject to quasi-static fading channels, which can degrade the retransmission performance. Moreover, temporal diversity cannot guarantee the reliability requirements of URLLC packet transmission if the channel exhibits deep fading over a long period of time. Therefore, it is crucial to investigate new technologies or develop new communication protocols to ensure ultra-reliable communication between IIoT devices within a certain millisecond delay.
\vspace{-3mm}
\subsection{Related Work}
Intelligent reflecting surface (IRS), a metasurface equipped with massive reflecting elements \cite{SurPC2021}, is a promising technology for enabling URLLC \cite{Cheng2021}. It is able to reconfigure the wireless environment and turn the random wireless channels into partially deterministic ones by beamforming design \cite{ZhaiIIoTJ2021,GuoIToC2021,CaiIToWC2022}. As a result, the received signal-to-noise ratio (SNR) can be significantly improved. Even though the direct link between transceivers is hindered, IRS can create a virtual line-of-sight (LoS) link to bypass obstacles between transceivers via smart reflection \cite{GongICST2020}. Thus, the integration of IRS into the communication system helps to enhance reliability, reduce packet retransmission, and minimize the delay. Consequently, IRS can be a potential and cost-effective solution to realize URLLC. In \cite{HashemiIToVT2021}, the authors presented the performance analysis of the average achievable rate and error probability over an IRS-aided URLLC transmission with/without phase noise. Considering non-linear energy harvesting, the end-to-end performance of the IRS-assisted wireless system was analyzed in \cite{DhokIToC2021} for industrial URLLC applications, and the approximate closed-form expression of block error rate was derived. Authors in \cite{XieIWCL2021} studied an IRS-assisted downlink multiuser URLLC system and jointly optimized the user grouping and the blocklength allocation at the base station (BS), as well as the reflective beamforming at the IRS for latency minimization.

Device-to-device (D2D) communication is another potential technology to achieve URLLC. In most mission-critical applications, devices (e.g. sensors, machines, actuators, robots) are in close proximity to each other. Thus, the channel between the devices is much more reliable than that between the access point (AP) and the device, thereby rendering the D2D network promising to reduce resource consumption, lower communication latency, and improve reliability \cite{JameelICST2018,AnsariISJ2018}. Recently, there are studies on the design of D2D-based URLLC systems. A probability-based D2D activation and power allocation scheme was proposed in \cite{ChangIToC2021} to deal with the extremely high quality-of-service (QoS) requirements in URLLC for real-time wireless control systems, where each sensor autonomously decides whether to participate in the control process without interactive communications. The authors in \cite{LiuITWC2018} developed a D2D-based two-phase transmission protocol for URLLC, where each group's messages are combined and the BS multicasts them to the leaders in groups in the first phase, while leaders help to relay messages to other users in their groups in the second phase. In \cite{WuIToVT2020}, authors investigated the contention-based radio resource management for URLLC-oriented D2D communications.

Some research works further combined the technologies of IRS and D2D communication and studied the IRS-assisted D2D network design. For instance, the authors in \cite{VanIToWC2022} considered deploying IRS in the integrated data and energy network coexisting with D2D communication to maximize the minimum throughput of the information-demanded users. In \cite{WangIToWC2021}, a resource allocation design for the IRS-aided joint processing coordinated multipoint (JP-CoMP) system with underlaying D2D network was investigated. The authors in \cite{PengICL2021} studied an IRS-aided D2D communication system over Rician fading channels with the  consideration of practical hardware impairments at both the terminals and IRSs. However, the optimization of most of these works is based on Shannon capacity with assumptions of infinite blocklength and zero error probability. If we directly apply the results and conclusion of these works to URLLC-oriented applications, we may get the underestimated delay performance and overestimated reliability performance \cite{XuIToWC2020}. This necessitates the IRS-assisted D2D network design for URLLC under the finite blocklength (FBL) regime. Towards this end, one practical factor needed to be considered is the channel state information (CSI) estimation. On the one hand, IRS can only passively reflect signals and is not able to transmit or receive pilot symbols. On the other hand, the transmission time interval (TTI) of URLLC systems is very short, so that the time for channel training is highly restricted. Consequently, the perfect knowledge of CSI may not be available in practice. This entails a robust IRS-assisted D2D network design for URLLC under the imperfect CSI scenario.
\vspace{-4mm}
\subsection{Contributions}
In this paper, we propose a two-stage protocol to enable reliable communication between the AP and actuators in the IIoT scenario assisted by IRS and D2D networks under the scenarios of perfect and imperfect CSI. This is achieved by jointly optimizing the active beamforming at the AP and the reflective beamforming at the IRS to maximize the number of actuators with successful decoding. The main contributions of this paper are summarized as follows.
\begin{itemize}
  \item A communication protocol for dual augmented reliability by combining IRS and D2D networks in a specified latency requirement is proposed. The reliable communication design is investigated under the scenarios of perfect and imperfect CSI. In this way, URLLC in IIoT scenario can be enabled.
  \item With perfect CSI, we propose two efficient algorithms with guaranteed convergence and polynomial time complexity. One is the AltMin algorithm, dealing with unit-modulus constraints by relaxation first and then projection into a feasible region. The other one is a penalty-based successive convex approximation (SCA) algorithm which decomposes the product of active and passive beamformers and avoids using the alternating optimization method. For the imperfect CSI case, a semidefinite relaxation (SDR)-based block coordinate descent (BCD) algorithm is proposed for robust design. The proposed algorithms show better performance than other baseline schemes.
  \item The advantages of IRS with well-optimized phase shifts and the D2D network to enhance reliable communication are verified by simulations compared to other baseline schemes. Due to the doubly improved reliability from the combined usage of IRS and D2D network, as well as the multiuser diversity, the proposed two-stage protocol can ensure reliable communication between AP and actuators even under stringent delay requirements and it is shown to be robust to the uncertainties of CSI.
\end{itemize}

The rest of the paper is organized as follows. In Section \ref{sysmodelsec}, we present the system model, the two-stage communication protocol, and the problem formulation. The reliable communication design under the perfect and imperfect CSI scenarios are given in Section \ref{speicalcase} and Section \ref{worstcasesec}, respectively. Simulation results are presented in Section \ref{simuresults} and Section \ref{conclusionsec} concludes the paper.

\textit{Notations:}~$\mathbf{I}_n$ refers to an $n\times n$ identity matrix. $\mathbf{x}_i,\mathbf{X}_{ij}$ stand for the $i$-th element of a vector $\mathbf{x}$ and the $(i,j)$-th element of a matrix $\mathbf{X}$, respectively. $\mathrm{diag}(\mathbf{x})$ denotes a diagonal matrix whose diagonal elements are extracted from a vector $\mathbf{x}$.
\section{System Model}\label{sysmodelsec}
We consider a downlink multiple-input single-output (MISO)-URLLC system between AP and $K$ actuators in the IIoT scenario, where AP is equipped with $N_t$ antennas and all actuators indexed by $k=\{1,\cdots,K\}$ are equipped with a single antenna. As shown in Fig. \ref{systemModel}, inspired by a D2D-based two-phase transmission protocol \cite{LiuITWC2018}, we leverage an IRS with $M$ reflecting elements and the D2D network to doubly enhance the transmission reliability. Thus, all actuators can successfully receive critical messages in the form of short packets from AP within a delay of $\tau$ seconds. The symbols used throughout the paper and their definitions are listed in Table \ref{symbolDef}.
\begin{figure}[htb]
\centering
\includegraphics[width=0.95\linewidth]{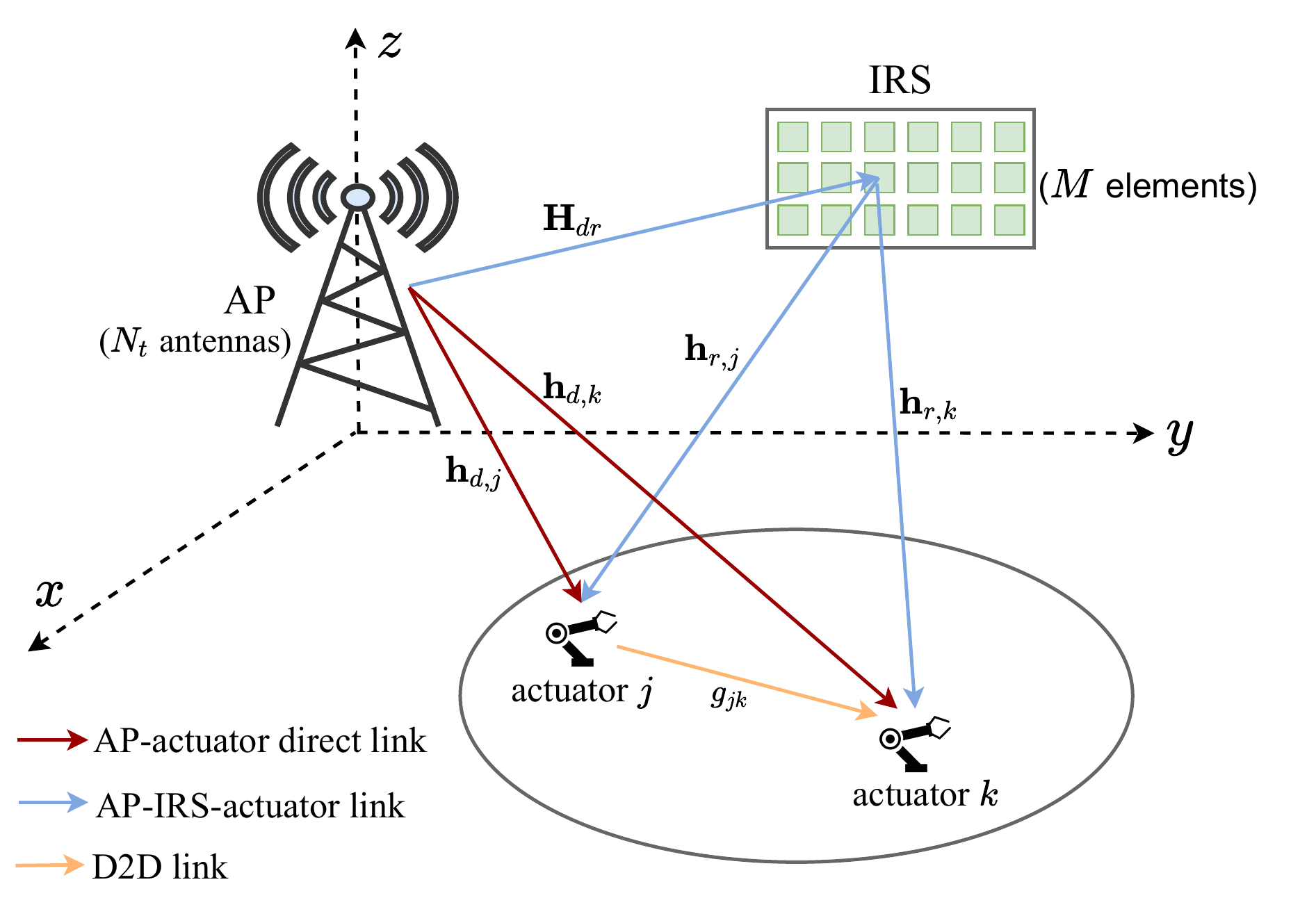}
	\caption{A downlink MISO-URLLC system between AP and multiple actuators where transmission reliability is doubly enhanced by IRS and D2D network.}
	\label{systemModel}
\end{figure}

\begin{table}[htb]
\centering
\caption{List of symbols.}
\label{symbolDef}
\begin{tabular}{cc}
\toprule
Symbol & Definition \\
\midrule
$K$ & Number of actuators\\
$N_t$ &Number of transmit antennas at AP  \\
$M$ & Number of reflecting elements at IRS\\
$\mathbf{h}_{d,k}$ & Direct channel from AP to actuator $k$ \\
$\mathbf{G}_k$ & Cascaded AP-IRS-actuator channel \\
$\mathbf{R}_k$ & Effective channel from AP to IRS \\
$\mathbf{w}$ & Active beamforming vector at AP \\
$\widetilde{\mathbf{v}}$ & Phase shift vector at IRS \\
$\mathbf{W}$ & Active beamforming matrix at AP\\
$\mathbf{V}$ & Reflective beamforming matrix at IRS \\
\bottomrule
\end{tabular}
\end{table}

The channels from the AP to the actuator $k$, from the AP to the IRS, and from the IRS to the actuator $k$ are denoted by $\mathbf{h}_{d,k}\in\mathbb{C}^{N_t},\mathbf{H}_{dr}\in\mathbb{C}^{M\times N_t},\mathbf{h}_{r,k}\in\mathbb{C}^{M}$, respectively. We denote the channel from the actuator $j$ to actuator $k$ by $g_{jk}\in\mathbb{C}$. The coherence time and bandwidth of all channels are assumed to be the same for tractability, denoted as $T$ seconds and $B$ Hz, respectively. Given that the positions of all actuators are fixed and the target latency requirement $\tau$ is much smaller than the channel coherence time $T$, we adopt a quasi-static block-fading channel model, i.e., the channel gains remain constant within one coherence block and vary independently across blocks.

The transmission from AP to the actuators follows a two-stage IRS-aided D2D communication protocol, as shown in Fig. \ref{two-stageprotocol}. The description of each stage is given below.
\begin{figure}[t]
    \setlength{\belowcaptionskip}{-5mm}
	\centering
	\includegraphics[width=\linewidth]{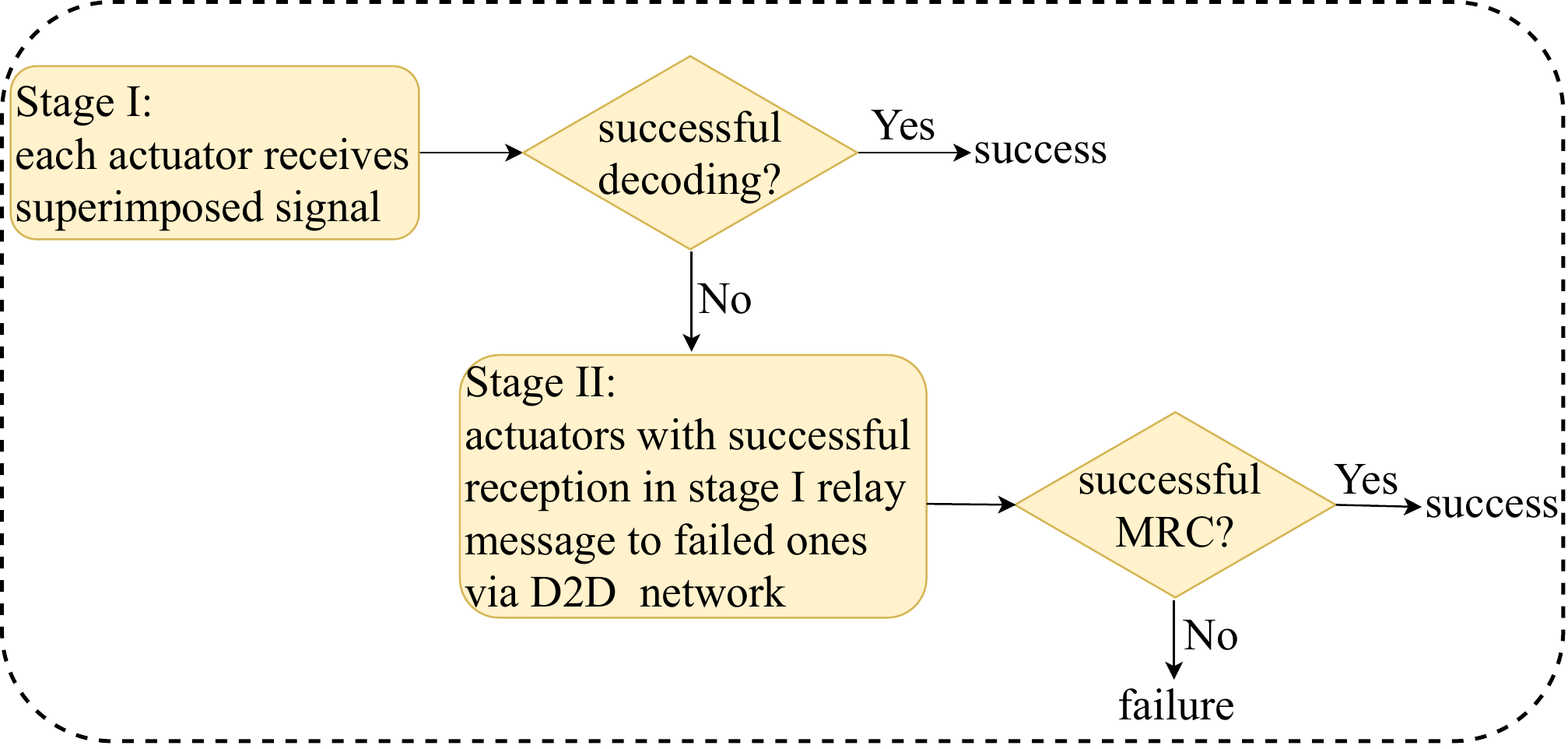}
	\caption{The workflow of the proposed two-stage protocol, where Stage I and Stage II are described in Section \ref{RISstageI} and Section \ref{D2DstageII}.}
	\label{two-stageprotocol}
\end{figure}
\vspace{-3mm}
\subsection{IRS-aided first-stage transmission}\label{RISstageI}
In the first stage with a duration of $\tau_1$ seconds, the transmitted signal at the AP is $\mathbf{x}^{(\text{I})}=\mathbf{w} s$, where $\mathbf{w}\in\mathbb{C}^{N_t}$ is the beamforming vector for broadcasting at the AP, and $s$ denotes the combined symbol with unit power intended for all actuators. The received signal at the actuator $k$ consists of a direct signal from AP and a reflected signal from IRS, which can be given by $y_k^{(\text{I})}=\left(\mathbf{h}_{d,k}^H+\mathbf{h}_{r,k}^H\pmb{\Phi}\mathbf{H}_{dr}\right)\mathbf{x}^{(\text{I})}
  +z_k^{(\text{I})}$, where $\pmb{\Phi}=\mathrm{diag}(\phi_1,\cdots,\phi_M)$ is a reflection matrix of IRS with $\phi_m=e^{j\theta_m},\forall m$ on it diagonal, $\theta_m\in[0,2\pi]$ is the phase shift of the $m$-th reflecting element, and $z_k^{(\text{I})}\sim\mathcal{CN}(0,\sigma_k^2)$ is the additive white Gaussian noise (AWGN).

Denote the cascaded AP-IRS-actuator channel and the vector containing diagonal elements of matrix $\pmb{\Phi}$ by $\mathbf{G}_k=\mathrm{diag}\left(\mathbf{h}_{r,k}^H\right)\mathbf{H}_{dr}$ and  $\widetilde{\mathbf{v}}=[\phi_1,\cdots,\phi_M]^H\in\mathbb{C}^{M}$, respectively. The SNR of the $k$-th actuator in the first stage can be represented as
\begin{equation}\label{gammaI}
\setlength\abovedisplayskip{2pt}
\setlength\belowdisplayskip{2pt}
  \gamma_k^{(\text{I})}=\frac{\begin{vmatrix}\left(\mathbf{h}_{d,k}^H+\widetilde{\mathbf{v}}^H\mathbf{G}_k\right)
  \mathbf{w}\end{vmatrix}^2}{\sigma_k^2}.
\end{equation}

Based on the normal approximation of the maximum achievable rate with finite blocklength codes \cite{YangIToIT2014}, we characterize the maximum achievable rate of the actuator $k$ over $L_1=\tau_1B$ channel uses in the first stage by
\begin{equation}
\setlength\abovedisplayskip{2pt}
\setlength\belowdisplayskip{2pt}
  \frac{D}{L_1}=\log_2\left(1+\gamma_k^{(\text{I})}\right)
  -\sqrt{\frac{V_k^{(\text{I})}}{L_1}}
  Q^{-1}\left(\varepsilon_k^{(\text{I})}\right),\label{FBL_stageI}
\end{equation}
where $D$ is the number of data bits, $V_k^{(\text{I})}=(\log_2 e)^2\left(1-\left(1+\gamma_k^{(\text{I})}\right)^{-2}\right)$ is the channel dispersion, $\varepsilon_k^{(\text{I})}$ is the decoding error probability, $Q^{-1}(\cdot)$ is the inverse of Q-function $Q(x)=\frac{1}{\sqrt{2 \pi}}\int_{x}^{\infty} e^{-t^{2}/2}dt$. Note that equation \eqref{FBL_stageI} characterizes the relationship between the achievable rate, transmission time (channel uses) and decoding error probability. From \eqref{FBL_stageI}, the decoding error probability can be represented as
\begin{equation}
\setlength\abovedisplayskip{2pt}
\setlength\belowdisplayskip{2pt}
  \varepsilon_k^{(\text{I})}\!=\!Q\left(\frac{\sqrt{L_1}\log_2\left(1
  \!+\!\gamma_k^{(\text{I})}\right)\!-\!D/\sqrt{L_1}}
  {\sqrt{V_k^{(\text{I})}}}\right)\!\triangleq \! Q\left(f_1\left(\gamma_k^{(\text{I})}\right)\right).\label{epsilon1}\\
\end{equation}

Note that $f_1\left(\gamma_k^{(\text{I})}\right)$ is a monotonically increasing function of $\gamma_k^{(\text{I})}$ \cite[Lemma 1]{MalakIToWC2019} and Q-function is a monotonically decreasing function. So it can be obtained that $\varepsilon_k^{(\text{I})}$ decreases with $\gamma_k^{(\text{I})}$. Denote the maximum allowed packet error probability (PEP) by $\varepsilon_{\text{th}}$, then we can have the minimum SNR threshold $\gamma_{\text{th}}^{(\text{I})}$ required for successful decoding in the first stage corresponding to the maximum PEP $\varepsilon_{\text{th}}$.

If the SNR of actuator $k$ can reach the threshold $\gamma_{\text{th}}^{(\text{I})}$, we regard that the $k$-th actuator can successfully decode the message in the first stage, otherwise, the packet reception of actuator $k$ ends with failure in the first stage. Thus, for any actuator $k$, we define an indicator $a_k^{(\text{I})}$ as follows
\begin{equation}\label{indicatorI}
\setlength\abovedisplayskip{2pt}
\setlength\belowdisplayskip{2pt}
    a_k^{(\text{I})}=\left\{\begin{array}{ll}
    1,& \mathrm{if}~ \gamma_k^{(\text{I})}\geq \gamma_{\text{th}}^{(\text{I})},\\
    0,& \mathrm{otherwise}.
    \end{array}\right.
\end{equation}
\vspace{-3mm}
\subsection{Second-stage D2D transmission}\label{D2DstageII}
In the second stage with a duration of $\tau_2=\tau-\tau_1$ seconds, the actuators with successful decoding in the first stage relay messages over a D2D network to the rest of actuators. The transmit signal of the actuator $j$ in the second stage is $x_j^{(\text{II})}=a_j^{(\text{I})}\sqrt{P}s$, where $P$ is the common transmit power for all actuators. Here, we don't consider power control due to the following two factors: 1) the global CSI of the D2D network is not available at the AP; 2) the actuators that act as relays transmit the signal at the maximum power, which can improve the received SNR and is more aligned with the requirements of the URLLC scenario.

If actuator $k$ cannot decode the message successfully in the first stage, it will receive relayed signals in the second stage from actuators with successful reception. Thus, the received signal of actuator $k$ in the second stage can be expressed as $y_k^{(\text{II})}=\sum\limits_{j\neq k}g_{jk}x_j^{(\text{II})}+z_k^{(\text{II})}$, where $z_k^{(\text{II})}\sim\mathcal{CN}(0,\sigma_k^2)$ is the AWGN. At the end of this stage, the actuator $k$ jointly decodes the signals received in two stages via maximum-ratio combining (MRC). The received SNR of actuator $k$ in the second stage is given by
\begin{equation}\label{gammaII}
\setlength\abovedisplayskip{2pt}
\setlength\belowdisplayskip{2pt}
  \gamma_k^{(\text{II})}=\gamma_k^{(\text{I})}+\frac{P
  \begin{vmatrix} \sum\limits_{j\neq k}a_j^{(\text{I})}g_{jk}\end{vmatrix}^2}{\sigma_k^2}.
\end{equation}
Then, the maximum achievable rate of actuator $k$ over $L_2=\tau_2B$ channel uses in the second stage can be characterized by
\begin{equation}
\setlength\abovedisplayskip{2pt}
\setlength\belowdisplayskip{2pt}
  \frac{D}{L_2}=\log_2\left(1+\gamma_k^{(\text{II})}\right)
  -\sqrt{\frac{V_k^{(\text{II})}}{L_2}}
  Q^{-1}\left(\varepsilon_k^{(\text{II})}\right),\label{FBL_stageII}
\end{equation}
where $\varepsilon_k^{(\text{II})}$ is the decoding error probability. From equation \eqref{FBL_stageII}, $\varepsilon_k^{(\text{II})}$ can be expressed as
\begin{equation}
\setlength\abovedisplayskip{4pt}
\setlength\belowdisplayskip{4pt}
  \!\varepsilon_k^{(\text{II})}\!=\!Q\left(\!\frac{\sqrt{L_2}\log_2\left(1
  \!+\!\gamma_k^{(\text{II})}\right)\!-\!D/\sqrt{L_2}}
  {\sqrt{V_k^{(\text{II})}}}\right)\!\triangleq \! Q\left(f_2\left(\gamma_k^{(\text{II})}\right)\right).\label{epsilon2}\\
\end{equation}
Similar to the conclusion derived from \eqref{epsilon1}, we know that $\varepsilon_k^{(\text{II})}$ decreases with $\gamma_k^{(\text{II})}$. Then, we can obtain the minimum SNR threshold $\gamma_{\text{th}}^{(\text{II})}$ for successful decoding in the second stage corresponding to the maximum PEP $\varepsilon_{\text{th}}$. To indicate whether actuator $k$ failed in the first stage can successfully decode the message in the second stage, we denote an indicator by
\begin{equation}\label{indicatorII}
\setlength\abovedisplayskip{2pt}
\setlength\belowdisplayskip{2pt}
    a_k^{(\text{II})}=\left\{\begin{array}{ll}
    1,& \mathrm{if}~ a_k^{(\text{I})}=0,\gamma_k^{(\text{II})}\geq \gamma_{\text{th}}^{(\text{II})},\\
    0,& \mathrm{if}~ a_k^{(\text{I})}=0,\gamma_k^{(\text{II})}< \gamma_{\text{th}}^{(\text{II})}.
    \end{array}\right.
\end{equation}
\vspace{-4mm}
\subsection{Problem Formulation}
In this paper, we aim to develop a reliable communication design between an AP and multiple actuators within a given delay requirement. Based on the two-stage communication protocol, our goal is to maximize the number of actuators with successful decoding by jointly designing the active beamforming at the AP and the reflection matrix at the IRS, i.e.,
\begin{subequations}
\begin{align}
\mathrm{P1}:\max_{\mathbf{w},\widetilde{\mathbf{v}}} \quad&\sum_{k=1}^K \left(a_k^{(\text{I})}+a_k^{(\text{II})}\right)\\
\mathrm{s.t.}\quad &\|\mathbf{w}\|^2\leq P_{\text{max}},\label{P1power}\\
&|\widetilde{\mathbf{v}}_m|=1,\forall m,\label{P1unitmodulus}
\end{align}
\end{subequations}
where constraints \eqref{P1power} and \eqref{P1unitmodulus} denote the maximum transmit power constraint at AP and the unit-modulus constraint of each reflecting element.

We assume that at least one actuator can successfully decode the message in the first stage. Note that the higher the number of actuators that can successfully decode the signal in the first stage (i.e., the higher the number of actuators that act as relays to relay messages in the second stage), the higher the probability of successful reception that the remaining actuators have via D2D network in the second stage \cite{LiuITWC2018}. This is mainly due to the multiuser diversity and the strong D2D link between actuators in close proximity. Motivated by this, we transform problem $\mathrm{P1}$ into an alternative form that aims to maximize the number of actuators with successful decoding in the first stage. As a result, fewer remaining actuators are required to reach the SNR threshold $\gamma_{\text{th}}^{(\text{II})}$ for reliable reception in the second stage. Therefore, problem $\mathrm{P1}$ can be reformulated as
\vspace{-1mm}
\begin{subequations}
\setlength\abovedisplayskip{4pt}
\setlength\belowdisplayskip{4pt}
\begin{align}
\mathrm{P2}:\max_{\mathbf{w},\widetilde{\mathbf{v}}} \quad&\sum_{k=1}^K a_k^{(\text{I})}\\
\mathrm{s.t.}\quad~ &\sum_{k=1}^K a_k^{(\text{I})}\geq 1,\forall k,\label{P2C1}\\
&\|\mathbf{w}\|^2\leq P_{\text{max}},\\
&|\widetilde{\mathbf{v}}_m|=1,\forall m. \label{P2C3}
\end{align}
\end{subequations}

We can obtain the active beamformer $\mathbf{w}$ at the AP, phase shifts $\widetilde{\mathbf{v}}$ and indicator $a_k^{(\text{I})},\forall k$ by solving problem $\mathrm{P2}$. With those solutions, we can calculate $\gamma_k^{(\text{I})}$ and $\gamma_k^{(\text{II})}$ based on equations \eqref{gammaI} and \eqref{gammaII}. Then, we can obtain $a_k^{(\text{II})}$ based on equation \eqref{indicatorII}. It is necessary to remark that if $\sum_{k=1}^K \begin{pmatrix}a_k^{(\text{I})}+a_k^{(\text{II})}\end{pmatrix}=K$, then all actuators can successfully receive the critical signals from AP after the two-stage transmission.
\vspace{-3mm}
\section{Active and Reflective Beamforming with Perfect CSI}\label{speicalcase}
In this section, to achieve a reliable communication design with perfect CSI, we propose two effective algorithms with guaranteed convergence, namely the AltMin algorithm and the penalty-based SCA algorithm.
\vspace{-4mm}
\subsection{Problem Transformation}
In this case, the CSI of direct AP-actuator channels $\mathbf{h}_{d,k},\forall k$ and cascaded AP-IRS-actuator channels $\mathbf{G}_k,\forall k$ are assumed to be perfectly known at the AP, which can be obtained by channel estimation methods in \cite{HuIToC2021a}.

In problem $\mathrm{P2}$, $a_k^{(\text{I})}$ given by equation \eqref{indicatorI} is a discrete function over beamformers $\mathbf{w}$ and $\widetilde{\mathbf{v}}$. This discrete nature complicates the solution and makes it difficult to apply the standard optimization techniques. In view of this, we introduce a set of auxiliary variables $\mathbf{q}=[q_1,\cdots,q_k]^T\succeq\mathbf{0}$ to characterize the SNR gap between the achievable SNR $\gamma_k^{(\text{I})}$ and the SNR threshold $\gamma_{\text{th}}^{(\text{I})}$. In particular, $q_k=0$ means the SNR condition $\gamma_k^{(\text{I})}\geq \gamma_{\text{th}}^{(\text{I})}$ is satisfied and $a_k^{(\text{I})}=1$. Otherwise, $a_k^{(\text{I})}=0$ for nonzero $q_k$. That is, the number of zero elements in vector $\mathbf{q}$ denotes the number of actuators with successful decoding. Therefore, maximizing $\sum_{k=1}^K a_k^{(\text{I})}$ (the number of actuators with successful decoding in the first stage) is equivalent to minimizing $\|\mathbf{q}\|_0$ (the number of nonzero elements in vector $\mathbf{q}$). In this way, we transform $\mathrm{P2}$ into an equivalent but a more tractable and continuous form, which can be formulated as
\begin{subequations}
\setlength\abovedisplayskip{2pt}
\setlength\belowdisplayskip{2pt}
\begin{align}
\mathrm{P3}:\min_{\mathbf{q,w},\widetilde{\mathbf{v}}} \quad&\|\mathbf{q}\|_0\\
\mathrm{s.t.}~~~~ &\gamma_k^{(\text{I})}+q_k\geq \gamma_{\text{th}}^{(\text{I})},\forall k,\label{P3C1}\\
&q_k\geq 0,\forall k,\label{P3C2}\\
&\|\mathbf{q}\|_0\leq K-1,\label{P3C3}\\
&\|\mathbf{w}\|^2\leq P_{\text{max}},\label{powerConstraint}\\
&|\widetilde{\mathbf{v}}_m|=1,\forall m,\label{unit-modulus}
\end{align}
\end{subequations}
where \eqref{P3C3}, equivalent to \eqref{P2C1}, guarantees that at least one actuator can successfully decode the message  in the first stage.

The challenges in solving $\mathrm{P3}$ mainly include the $\ell_0$-norm in the objective function and constraint \eqref{P3C3}, the strictly nonconvex constraint \eqref{P3C1} due to the coupling of $\mathbf{w}$ and $\widetilde{\mathbf{v}}$ in the form of a product, and the unit-modulus constraint \eqref{unit-modulus}. To overcome these challenges, we propose two effective algorithms below.
\vspace{-3mm}
\subsection{Proposed AltMin Algorithm}\label{AltMinSection}
In this algorithm, we apply an alternating optimization technique to solve this minimization problem. Hence, we call this algorithm AltMin. Specifically, we alternately optimize $\mathbf{q,w}$ with given $\widetilde{\mathbf{v}}$ and $\mathbf{q},\widetilde{\mathbf{v}}$ with given $\mathbf{w}$. In the following, we elaborate the design of the AltMin algorithm.

With given $\widetilde{\mathbf{v}}$, let $\mathbf{b}_k=\mathbf{h}_{d,k}+\mathbf{G}_k^H\widetilde{\mathbf{v}}$, then constraint \eqref{P3C1} can be recast as
\begin{equation}
\setlength\abovedisplayskip{2pt}
\setlength\belowdisplayskip{2pt}
  |\mathbf{b}_k^H\mathbf{w}|^2\geq \sigma_k^2\left(\gamma_{\text{th}}^{(\text{I})}-q_k\right),\forall k.
\end{equation}
Further, based on the first-order Taylor approximation, $|\mathbf{b}_k^H\mathbf{w}|^2$ is lower bounded by
\begin{equation}
\setlength\abovedisplayskip{2pt}
\setlength\belowdisplayskip{2pt}
  \mathcal{F}_k(\mathbf{w})=2\mathrm{Re}\begin{Bmatrix}\left(\mathbf{w}^{(i)}\right)^H\mathbf{b}_k\mathbf{b}_k^H
  \mathbf{w}\end{Bmatrix}-|\mathbf{b}_k^H\mathbf{w}^{(i)}|^2,
\end{equation}
where $\mathbf{w}^{(i)}$ is the solution obtained in the $i$-th iteration. In this way, \eqref{P3C1} can be approximated by a convex constraint, i.e.,
\begin{equation}
\setlength\abovedisplayskip{2pt}
\setlength\belowdisplayskip{2pt}
  \mathcal{F}_k(\mathbf{w})\geq \sigma_k^2\left(\gamma_{\text{th}}^{(\text{I})}-q_k\right),\forall k.
\end{equation}

Next, we apply the reweighted $\ell_1$ technique \cite{CandesJoFAA2007} to deal with the $\ell_0$-norm. The principle is to use the weighted $\ell_1$-norm to approximate the $\ell_0$-norm and then update the weights. To be more specific, $\|\mathbf{q}\|_0$ can be approximated by $\sum_{k=1}^K \omega_kq_k$, where $\omega_k,\forall k$ are positive weights. Here, weights are updated by $\omega_k^{(i+1)}=1/\left(q_k^{(i)}+\nu\right)$, where $q_k^{(i)}$ is the solution obtained in the $i$-th iteration and $\nu$ is a small number. Note that the weights can be interpreted as penalties, i.e., large weights encourage zero entry $q_k$ while small weights discourage nonzero $q_k$.

Based on the above transformations, the optimization problem of $\mathbf{q,w}$ with given $\widetilde{\mathbf{v}}$ can be given by
\vspace{-1mm}
\begin{subequations}
\setlength\abovedisplayskip{2pt}
\setlength\belowdisplayskip{2pt}
\begin{align}
\!\!\!\!\mathrm{AltMin\!-\!P4\!-\!1}:\min_{\mathbf{q,w}} \quad&\sum_{k=1}^K \omega_kq_k\\
\mathrm{s.t.}~~~ &\mathcal{F}_k(\mathbf{w})\geq \sigma_k^2\left(\gamma_{\text{th}}^{(\text{I})}-q_k\right),\forall k,\\
&q_k\geq 0,\forall k,\\
&\sum_{k=1}^K \omega_kq_k\leq K-1,\\
&\|\mathbf{w}\|^2\leq P_{\text{max}}.
\end{align}
\end{subequations}
It is a convex problem, which can be solved by existing convex program solvers like CVX \cite{grant2014cvx}.

With given $\mathbf{w}$, constraint \eqref{P3C1} can be rewritten as
\begin{equation}
\setlength\abovedisplayskip{2pt}
\setlength\belowdisplayskip{2pt}
  |c_k+\widetilde{\mathbf{v}}^H\mathbf{d}_k|^2\geq \sigma_k^2\left(\gamma_{\text{th}}^{(\text{I})}-q_k\right),\forall k,
\end{equation}
where $c_k=\mathbf{h}_{d,k}^H\mathbf{w},\mathbf{d}_k=\mathbf{G}_k\mathbf{w}$. By applying the first-order Taylor approximation, $|c_k+\widetilde{\mathbf{v}}^H\mathbf{d}_k|^2$ is lower bounded by
\begin{align}
\setlength\abovedisplayskip{2pt}
\setlength\belowdisplayskip{2pt}
  \mathcal{G}_k(\widetilde{\mathbf{v}})=&2\mathrm{Re}\begin{Bmatrix}\left(c_k\mathbf{d}_k^H
  +\left(\widetilde{\mathbf{v}}^{(i)}\right)^H
  \mathbf{d}_k\mathbf{d}_k^H\right)\left(\widetilde{\mathbf{v}}-\widetilde{\mathbf{v}}^{(i)}\right)
  \end{Bmatrix}\notag\\
  &+\begin{vmatrix}c_k+\left(\widetilde{\mathbf{v}}^{(i)}\right)^H\mathbf{d}_k
  \end{vmatrix}^2,
\end{align}
where $\widetilde{\mathbf{v}}^{(i)}$ is the solution obtained in the $i$-th iteration. Thus, we can convert constraint \eqref{P3C1} into a convex one
\begin{equation}
\setlength\abovedisplayskip{2pt}
\setlength\belowdisplayskip{2pt}
  \mathcal{G}_k(\widetilde{\mathbf{v}})\geq \sigma_k^2\left(\gamma_{\text{th}}^{(\text{I})}-q_k\right),\forall k.
\end{equation}

The same technique is applied to deal with $\|\mathbf{q}\|_0$. The remaining challenge is the nonconvex unit-modulus constraint \eqref{unit-modulus}. To make it more tractable, we first relax it as
\begin{equation}
\setlength\abovedisplayskip{2pt}
\setlength\belowdisplayskip{2pt}
  |\widetilde{\mathbf{v}}_m|\leq1,\forall m.
\end{equation}
Accordingly, we have the following convex optimization problem
\begin{subequations}
\setlength\abovedisplayskip{2pt}
\setlength\belowdisplayskip{2pt}
\begin{align}
\mathrm{AltMin\!-\!P4\!-\!2}:\min_{\mathbf{q,\widetilde{\mathbf{v}}}} \quad&\sum_{k=1}^K \omega_kq_k\\
\mathrm{s.t.}~~~ &\mathcal{G}_k(\widetilde{\mathbf{v}})\geq \sigma_k^2\left(\gamma_{\text{th}}^{(\text{I})}-q_k\right),\forall k,\\
&q_k\geq 0,\forall k,\\
&\sum_{k=1}^K \omega_kq_k\leq K-1,\\
&|\widetilde{\mathbf{v}}_m|\leq1,\forall m.
\end{align}
\end{subequations}
Note that the solution obtained by solving $\mathrm{AltMin\!-\!P4\!-\!2}$ may not satisfy the unit-modulus constraint. One way to construct feasible phase shifts is given by
\begin{equation}
\setlength\abovedisplayskip{2pt}
\setlength\belowdisplayskip{2pt}
  \widetilde{\mathbf{v}}^{*}_m=\widetilde{\mathbf{v}}_m/|\widetilde{\mathbf{v}}_m|,\forall m.
\end{equation}
With given $\mathbf{w}$ and the obtained solution $\widetilde{\mathbf{v}}^{*}$, the remaining $q_k$ can be obtained by the closed-form expression $q_k=\gamma_{\text{th}}^{(\text{I})}-\gamma_k^{(\text{I})}$.

The overall AltMin algorithm is summarized in Algorithm \ref{AltMin}. This algorithm alternately optimizes two blocks defined by problems $\mathrm{AltMin\!-\!P4\!-\!1}$ and $\mathrm{AltMin\!-\!P4\!-\!2}$, respectively. The objective function denoted by $\mathcal{A}=\sum_{k=1}^K \omega_kq_k$ follows that $\mathcal{A}(\mathbf{w}^{(i)},\widetilde{\mathbf{v}}^{(i)})\geq
\mathcal{A}(\mathbf{w}^{(i+1)},\widetilde{\mathbf{v}}^{(i)})\geq
\mathcal{A}(\mathbf{w}^{(i+1)},\widetilde{\mathbf{v}}^{(i+1)})$, where $\mathbf{w}^{(i+1)},\widetilde{\mathbf{v}}^{(i+1)}$ are the optimal solutions obtained
for both blocks. The inequalities hold since $\gamma_k^{(\text{I})}$ can be maximized with optimal $\mathbf{w}^{(i+1)},\widetilde{\mathbf{v}}^{(i+1)}$, thereby reducing $q_k$ and lowering the objective function. Hence, the algorithm is guaranteed to converge to a stationary point of problem $\mathrm{P3}$ in polynomial time \cite{RazaviyaynSJO2013}.

In each iteration, the complexity of solving problem $\mathrm{AltMin\!-\!P4\!-\!1}$ is $\mathcal{O}((K+N_t)^3(2K+2))$ \cite{ImrePolik2010}, where $K+N_t$ is the number of optimization variables and $2K+2$ is the number of affine and convex constraints. Similarly, the complexity of solving problem $\mathrm{AltMin\!-\!P4\!-\!2}$ is $\mathcal{O}((K+M)^3(2K+M+1))$. Thus, the computational complexity order per iteration is $\mathcal{O}(2K(K+N_t)^3+(2K+M)(K+M)^3)$.
\begin{algorithm}[t]
\caption{\small \textbf{:} AltMin Algorithm for Solving Problem $\mathrm{P3}$ with Perfect CSI}
\label{AltMin}
\begin{spacing}{1.2}
	\small{\begin{algorithmic}[1]
		\STATE \textbf{Initialize} $\mathbf{w}^{(0)},\widetilde{\mathbf{v}}^{(0)}$, convergence tolerance $\epsilon$, and set outer iteration index $t=0$.
        \STATE Use bisection search method to obtain $\gamma_{\text{th}}^{(\text{I})}$.
		\REPEAT
        \STATE Given $\widetilde{\mathbf{v}}^{(t)}$, inner iteration index $i=0$
        \REPEAT
        \STATE Update $\mathbf{q}^{(i+1)},\mathbf{w}^{(i+1)}$ by solving problem $\mathrm{AltMin\!-\!P4\!-\!1}$ with given $\mathbf{w}^{(i)}$.
        \STATE $i\leftarrow i+1$.
        \UNTIL $|\mathcal{A}_1^{(i)}-\mathcal{A}_1^{(i-1)}|\leq \epsilon$.
        \STATE Set $\mathbf{q}^{(t)}=\mathbf{q}^{(i)},\mathbf{w}^{(t+1)}=\mathbf{w}^{(i)}$.
        \STATE Given $\mathbf{w}^{(t+1)}$ , inner iteration index $j=0$.
        \REPEAT
        \STATE Update $\mathbf{q}^{(j+1)},\widetilde{\mathbf{v}}^{(j+1)}$ by solving problem $\mathrm{AltMin\!-\!P4\!-\!2}$ with given $\widetilde{\mathbf{v}}^{(j)}$.	
        \STATE $j\leftarrow j+1$.	
        \UNTIL $|\mathcal{A}_2^{(j)}-\mathcal{A}_2^{(j-1)}|\leq \epsilon$.
        \STATE Set $\widetilde{\mathbf{v}}^{(j)}_m=\widetilde{\mathbf{v}}^{(j)}_m/|\widetilde{\mathbf{v}}^{(j)}_m|,\forall m$.
        \STATE Set $\mathbf{q}^{(t+1)}=\mathbf{q}^{(j)},\widetilde{\mathbf{v}}^{(t+1)}=\widetilde{\mathbf{v}}^{(j)}$.
        \STATE $t\leftarrow t+1$.
        \UNTIL $\|\mathbf{q}^{(t+1)}\|_1-\|\mathbf{q}^{(t}\|_1\leq \epsilon$.
	\end{algorithmic}}
\end{spacing}
\end{algorithm}
\vspace{-5mm}
\subsection{Proposed Penalty-based SCA Algorithm}\label{SCAsection}
In this algorithm, we first handle the coupling problem. Let $\mathbf{R}_k=\left[\mathbf{G}_k^H\quad \mathbf{h}_{d,k}\right]\in\mathbb{C}^{N_t\times (M+1)}$, and $\mathbf{v}=[\widetilde{\mathbf{v}}^T \quad1]^T$, then we can equivalently rewrite constraint \eqref{P3C1} as
\begin{equation}\label{SNRreformu}
\setlength\abovedisplayskip{3pt}
\setlength\belowdisplayskip{3pt}
  \mathrm{Tr}(\mathbf{W}\mathbf{R}_k\mathbf{V}\mathbf{R}_k^H)
  \geq\sigma_k^2\left(\gamma_{\text{th}}^{(\text{I})}-q_k\right),\forall k,
\end{equation}
where $\mathbf{W}=\mathbf{w}\mathbf{w}^H,\mathbf{V}=\mathbf{v}\mathbf{v}^H$. Note that $\mathbf{W}$ and $\mathbf{V}$ remain coupled in the form of a product in constraint \eqref{SNRreformu}. However, we can leverage the following equality to decompose the product, which is given by
\begin{equation}
\setlength\abovedisplayskip{3pt}
\setlength\belowdisplayskip{3pt}
  \langle\mathbf{A,B}\rangle_F=\frac{1}{2}\left(\|\mathbf{A}+\mathbf{B}\|_{F}^{2}-\|\mathbf{A}\|_{F}^{2}
  -\|\mathbf{B}\|_{F}^{2}\right),
\end{equation}
where $\langle\mathbf{A,B}\rangle_F=\mathrm{Tr}(A^HB)$ is the Frobenius inner product. Therefore, constraint \eqref{SNRreformu} can be equivalently written as
\begin{align}\label{Frobenreform}
\setlength\abovedisplayskip{3pt}
\setlength\belowdisplayskip{3pt}
  &\frac{1}{2}\left(\|\mathbf{W}+\mathbf{R}_k\mathbf{V}\mathbf{R}_k^H\|_{F}^{2}
  -\|\mathbf{W}\|_{F}^{2}
  -\|\mathbf{R}_k\mathbf{V}\mathbf{R}_k^H\|_{F}^{2}\right)\notag\\
  &~~~~~~~~~~~~~~~~~~~~~~~~~~~~~~~~~~~~\geq\sigma_k^2\left(\gamma_{\text{th}}^{(\text{I})}-q_k\right),\forall k.
\end{align}
Note that the left-hand-side (LHS) term in constraint \eqref{Frobenreform} is in the form of difference-of-convex function. We apply the first-order Taylor approximation to the convex function $\|\mathbf{W}+\mathbf{R}_k\mathbf{V}\mathbf{R}_k^H\|_{F}^{2}$, thus the LHS term is lower bounded by \eqref{equation25} as shown at the top of the next page,
\begin{figure*}[!t]
\setcounter{equation}{24}
\begin{align}
\label{equation25}
  &\mathcal{H}_k(\mathbf{W},\mathbf{V})=\frac{1}{2}\left(\|\mathbf{W}^{(i)}
  +\mathbf{R}_k\mathbf{V}^{(i)}\mathbf{R}_k^H\|_{F}^{2}+2\mathrm{Tr}\begin{pmatrix}\mathrm{Re}
  \begin{Bmatrix}\begin{pmatrix}\mathbf{W}^{(i)}
  +\mathbf{R}_k\mathbf{V}^{(i)}\mathbf{R}_k^H\end{pmatrix}\begin{pmatrix}\mathbf{W}
  -\mathbf{W}^{(i)}\end{pmatrix}\end{Bmatrix}\end{pmatrix}\right.\notag\\
  &\left.+2\mathrm{Tr}\begin{pmatrix}\mathrm{Re}
  \begin{Bmatrix}\begin{pmatrix}\mathbf{R}_k^H\begin{pmatrix}\mathbf{W}^{(i)}
  \!+\!\mathbf{R}_k\mathbf{V}^{(i)}
  \mathbf{R}_k^H\end{pmatrix}\mathbf{R}_k\end{pmatrix}
  \begin{pmatrix}\mathbf{V}
  \!-\!\mathbf{V}^{(i)}\end{pmatrix}\end{Bmatrix}\end{pmatrix}\!-\!
  \|\mathbf{W}\|_{F}^{2}
  \!-\!\|\mathbf{R}_k\mathbf{V}\mathbf{R}_k^H\|_{F}^{2}\right),
\end{align}
\hrulefill
\vspace*{1pt}
\end{figure*}
where $\mathbf{W}^{(i)},\mathbf{V}^{(i)}$ are the solutions obtained in the $i$-th iteration. In this way, we can transform constraint \eqref{Frobenreform} into a convex one
\begin{equation}\label{convexSNR}
\setlength\abovedisplayskip{3pt}
\setlength\belowdisplayskip{3pt}
  \mathcal{H}_k(\mathbf{W},\mathbf{V})\geq\sigma_k^2\left(\gamma_{\text{th}}^{(\text{I})}-q_k\right),\forall k.
\end{equation}

For the $\|\mathbf{q}\|_0$, we use reweighted $\ell_1$-norm for approximation. Through above transformations, problem $\mathrm{P3}$ can be reformulated as
\begin{subequations}
\setlength\abovedisplayskip{2pt}
\setlength\belowdisplayskip{2pt}
\begin{align}
\!\!\!\!\mathrm{SCA\!-\!P4}:\min_{\mathbf{q,W,V}} \quad&\sum_{k=1}^K \omega_kq_k\\
\mathrm{s.t.}\quad~ &\mathcal{H}_k(\mathbf{W},\mathbf{V})\geq\sigma_k^2\left(\gamma_{\text{th}}^{(\text{I})}-q_k\right),\forall k,\label{SCAP4C1}\\
&q_k\geq 0,\forall k,\\
&\sum_{k=1}^K \omega_kq_k\leq K-1,\\
&\mathrm{Tr}(\mathbf{W})\leq P_{\text{max}},\\
&\mathbf{V}_{mm}=1,m=1,\cdots,M+1,\label{diagV}\\
&\mathbf{W}\succeq \mathbf{0},\mathbf{V}\succeq \mathbf{0},\label{PSD}\\
&\mathrm{rank}(\mathbf{W})=1,\mathrm{rank}(\mathbf{V})=1,\label{rankone}
\end{align}
\end{subequations}
where constraint \eqref{diagV} is introduced to guarantee the unit-modulus of each IRS element.

Our main focus now is to deal with the rank-one constraints in \eqref{rankone}. To this end, we apply the following lemma to represent them equivalently.
\begin{lemma}\label{rankonelemma}
 For any $\mathbf{A}\in\mathbb{H}^{n}$, the constraint that $\mathbf{A}$ is rank-one is equivalent to
  \begin{equation}
\setlength\abovedisplayskip{2pt}
\setlength\belowdisplayskip{2pt}
    \|\mathbf{A}\|_{*}- \|\mathbf{A}\|_2\leq0.\label{Nuclear_spectal}
  \end{equation}
\end{lemma}
\begin{proof}
  For any $\mathbf{A}\in\mathbb{H}^{n}$, the inequality  $\|\mathbf{A}\|_{*}=\sum_{i}\sigma_i\geq\|\mathbf{A}\|_2=\max_{i}\{\sigma_i\}$ holds, where $\sigma_i$ is the $i$-th singular value of matrix $\mathbf{A}$. And the equality holds if and only if $\mathbf{A}$ is rank-one. Thus, the implicit constraint of Hermitian matrix $\mathbf{A}$, i.e., $\|\mathbf{A}\|_{*}- \|\mathbf{A}\|_2\geq0$ functions simultaneously with the constraint \eqref{Nuclear_spectal}, which requires that $\|\mathbf{A}\|_{*}- \|\mathbf{A}\|_2=0$, i.e., $\mathbf{A}$ is a rank-one matrix.
\end{proof}

According to Lemma \ref{rankonelemma}, rank-one constraints of $\mathbf{W,V}$ can be equivalently reformulated as
\begin{equation}
\setlength\abovedisplayskip{2pt}
\setlength\belowdisplayskip{2pt}
  \|\mathbf{W}\|_{*}-\|\mathbf{W}\|_2\leq 0,~\|\mathbf{V}\|_{*}-\|\mathbf{V}\|_2\leq 0.\label{nuclearSpectral}
\end{equation}
Then, we adopt the penalty-based method by moving constraint \eqref{nuclearSpectral} into the objective function, thereby resulting in the following optimization problem
\setlength\abovedisplayskip{2pt}
\setlength\belowdisplayskip{2pt}
\begin{align}
\!\!\!\!\mathrm{SCA\!-\!P5}:\min_{\mathbf{q,W,V}} \quad&\sum_{k=1}^K \omega_kq_k\!\!+\!\!\frac{1}{\rho}(\|\mathbf{W}\|_{*}\!\!-\!\!\|\mathbf{W}\|_2\!\!+\!\!\|\mathbf{V}\|_{*}\!\!
-\!\!\|\mathbf{V}\|_2)\\
\mathrm{s.t.}\quad~ &\eqref{SCAP4C1}-\eqref{PSD},\notag
\end{align}
where $\rho>0$ is a penalty factor penalizing the violation of constraint \eqref{nuclearSpectral}. The following proposition demonstrates the equivalence of problems $\mathrm{SCA\!-\!P4}$ and $\mathrm{SCA\!-\!P5}$. By solving problem $\mathrm{SCA\!-\!P5}$, we can recover $\mathbf{w},\mathbf{v}$ from rank-one solutions $\mathbf{W},\mathbf{V}$ by eigenvalue decomposition.

\begin{proposition}\label{Proposition1}
  Let $\mathbf{W}_s,\mathbf{V}_s$ be the optimal solutions of problem $\mathrm{SCA\!-\!P5}$ with penalty factor
$\rho_s$. For sufficiently small $\rho_s\rightarrow0$, any limit points $\overline{\mathbf{W}},\overline{\mathbf{V}}$ of the sequence $\{\mathbf{W}_s,\mathbf{V}_s\}$ are optimal solutions of problem $\mathrm{SCA\!-\!P4}$.
\end{proposition}
\proof Please refer to Appendix \ref{appendixA}.

Note that nuclear norm $\|\mathbf{W}\|_{*},\|\mathbf{V}\|_{*}$ and spectral norm $\|\mathbf{W}\|_2,\|\mathbf{V}\|_2$ are convex functions \cite{Chi2017}. Hence, the penalty terms are in a form of difference-of-convex function. To make it tractable, we exert the first-order Taylor series approximation on $\|\mathbf{W}\|_2$ and $\|\mathbf{V}\|_2$, thereby approximating the penalty terms as equations \eqref{equation31} and \eqref{equation32},
\begin{figure*}[!t]
\setcounter{equation}{30}
\begin{align}
  \mathcal{P}_w&\triangleq\frac{1}{\rho}\left(\|\mathbf{W}\|_{*}-\|\mathbf{W}^{(i)}\|_2-
  \mathrm{Tr}\begin{pmatrix}\mathrm{Re}\!\begin{Bmatrix}
  \pmb{\lambda}_{\max}(\mathbf{W}^{(i)})\pmb{\lambda}_{\max}^H(\mathbf{W}^{(i)})(\mathbf{W}-
  \mathbf{W}^{(i)})\end{Bmatrix}\!\end{pmatrix}\right),\label{equation31}\\
  \mathcal{P}_v&\triangleq\frac{1}{\rho}\left(\|\mathbf{V}\|_{*}-\|\mathbf{V}^{(i)}\|_2-
  \mathrm{Tr}\begin{pmatrix}\mathrm{Re}\!\begin{Bmatrix}
  \pmb{\lambda}_{\max}(\mathbf{V}^{(i)})\pmb{\lambda}_{\max}^H(\mathbf{V}^{(i)})(\mathbf{V}-
  \mathbf{V}^{(i)})\end{Bmatrix}\!\end{pmatrix}\right),\label{equation32}
\end{align}
\hrulefill
\vspace*{1pt}
\end{figure*}
where $\pmb{\lambda}_{\max}(\mathbf{W}),\pmb{\lambda}_{\max}(\mathbf{V})$ denote the eigenvector corresponding to the largest eigenvalue of matrix $\mathbf{W,V}$, respectively, and $\mathbf{W}^{(i)},\mathbf{V}^{(i)}$ are the solutions obtained in the $i$-th iteration.

Therefore, the optimization problem that needs to be solved in the $i$-th iteration can be expressed as
\setlength\abovedisplayskip{2pt}
\setlength\belowdisplayskip{2pt}
\begin{align}
\mathrm{SCA\!-\!P6}:\min_{\mathbf{q,W,V}} \quad&\sum_{k=1}^K \omega_kq_k+\mathcal{P}_w+\mathcal{P}_v \\
\mathrm{s.t.}\quad~ &\eqref{SCAP4C1}-\eqref{PSD},\notag
\end{align}
which can be efficiently solved by standard convex program solvers such as CVX \cite{grant2014cvx}.

According to the preceding analysis, we propose a penalty-based SCA algorithm to solve problem $\mathrm{SCA\!-\!P6}$, which is summarized in Algorithm \ref{PCSIalgorithm}. By iteratively solving a sequence of problem $\mathrm{SCA\!-\!P6}$ at feasible points, we can obtain a stationary-point solution of problem $\mathrm{P3}$ after convergence \cite{RazaviyaynSJO2013}.
\begin{algorithm}[t]
\caption{\small \textbf{:} Penalty-based SCA Algorithm for Solving Problem $\mathrm{P3}$ with Perfect CSI}
\label{PCSIalgorithm}
\begin{spacing}{1.2}
	\small{\begin{algorithmic}[1]
		\STATE \textbf{Initialize}~iteration index $i=0$, convergence tolerance $\epsilon$, feasible $\mathbf{q}^{(0)},\mathbf{W}^{(0)},\mathbf{V}^{(0)}$.
        \STATE Use bisection search method to obtain $\gamma_{\text{th}}^{(\text{I})}$.
        \STATE Calculate $\pmb{\lambda}_{\max}(\mathbf{W}^{(0)}),\pmb{\lambda}_{\max}(\mathbf{V}^{(0)})$.
		\REPEAT
        \STATE Set $i=i+1$.
		\STATE Obtain $\{\mathbf{q,W,V}\}$ by solving $\mathrm{SCA\!-\!P6}$.
        \STATE Update $\mathbf{q}^{(i)}=\mathbf{q},\mathbf{W}^{(i)}=\mathbf{W},\mathbf{V}^{(i)}=\mathbf{V}$.
        \STATE Update $\pmb{\lambda}_{\max}(\mathbf{W}^{(i)}),\pmb{\lambda}_{\max}(\mathbf{V}^{(i)})$.		
        \UNTIL convergence.
	\end{algorithmic}}
\end{spacing}
\end{algorithm}

For problem $\mathrm{SCA\!-\!P6}$, there are in total $N_t^2+(M+1)^2+K$ optimization variables and $2K+M+5$ affine and convex constraints. Therefore, the computational complexity order per iteration is $\mathcal{O}((N_t^2+M^2+K)^3(2K+M))$ \cite{ImrePolik2010}.
\section{Worst-case Robust Beamforming}\label{worstcasesec}
In this section, we proceed to investigate a robust design to achieve reliable communication in the worst-case. To resolve the robust beamforming problem, we propose a SDR-based BCD algorithm. Moreover, convergence and complexity analysis is given for the proposed algorithm.
\vspace{-4mm}
\subsection{Uncertainty Modeling}
Although the CSI of the direct channel $\mathbf{h}_{d,k}$ from the AP to the actuator $k$ can be estimated by conventional methods, CSI acquisition of AP-IRS channel $\mathbf{H}_{dr}$ and IRS-actuator channel $\mathbf{h}_{r,k}$ is more challenging. The is because without the help of active RF chains, a passive IRS cannot transmit and receive signals for channel estimation. Instead of estimating these two individual channels $\mathbf{H}_{dr}$ and $\mathbf{h}_{r,k}$, it is more preferable to attain the estimated CSI of the cascaded AP-IRS-actuator channel $\mathbf{G}_k$. In addition, the difficulty of channel estimation increases in URLLC systems since the time for pilot training is highly limited due to short TTI. Thus, it is more practical to consider the case with imperfect CSI of both direct AP-actuator channels and cascaded AP-IRS-actuator channels. This is different from most existing IRS-related works that only consider the uncertainty of cascaded AP-IRS-actuator channels \cite{Zhou2020,Wang2020}.

Channel uncertainty modeling mainly falls into two categories: 1) stochastic CSI error, where the CSI error is assumed to follow the circularly symmetric complex Gaussian (CSCG) distribution;  2) bounded CSI error, where all CSI errors lie in a norm-bounded uncertainty region. From \cite{HanIToWC2021}, we know that based on the worst-case bounding approach, we can convert the probabilistic constraint involving Gaussian CSI uncertainties into a deterministic norm-bounded form. Hence, in this paper, we adopt the norm-bounded CSI error model. Specifically, the CSI of the direct AP-actuator channel $\mathbf{h}_{d,k}$ and the cascaded AP-IRS-actuator channel $\mathbf{G}_k$ can be modeled as
\begin{align}
\setlength\abovedisplayskip{2pt}
\setlength\belowdisplayskip{2pt}
&\mathbf{h}_{d,k}=\hat{\mathbf{h}}_{d,k}+\Delta\mathbf{h}_{d,k},\notag\\
&\Delta\mathbf{h}_{d,k}\in
\Omega_{\mathbf{h}_{d,k}}\triangleq\{\Delta\mathbf{h}_{d,k}:\|\Delta\mathbf{h}_{d,k}\|_2
\leq \delta_{\mathbf{h}_{d,k}}\},\forall k,\notag\\
&\mathbf{G}_k=\widehat{\mathbf{G}}_k+\Delta\mathbf{G}_k,\notag\\
&\Delta\mathbf{G}_k\in\Omega_{\mathbf{G}_k}
\triangleq\{\Delta\mathbf{G}_k:\|\Delta\mathbf{G}_k\|_F\leq \delta_{\mathbf{G}_k}\},\forall k,
\end{align}
where $\hat{\mathbf{h}}_{d,k},\widehat{\mathbf{G}}_k$ are the estimates of the direct AP-actuator channel and the cascaded AP-IRS-actuator channel, respectively. The corresponding estimation errors $\Delta\mathbf{h}_{d,k},\Delta\mathbf{G}_k$, lying in continuous sets $\Omega_{\mathbf{h}_{d,k}},\Omega_{\mathbf{G}_k}$ are norm-bounded by $\delta_{\mathbf{h}_{d,k}},\delta_{\mathbf{G}_k}$, respectively. Meanwhile, $\delta_{\mathbf{h}_{d,k}}>0,\delta_{\mathbf{G}_k}>0$ represent the level of channel uncertainty, which can be chosen smaller for more accurate quantization and channel estimation.
\vspace{-5mm}
\subsection{Proposed SDR-based BCD Algorithm}
Under the case with imperfect CSI, $\mathbf{R}_k=\left[\mathbf{G}_k^H\quad \mathbf{h}_{d,k}\right]$, the effective channel from AP to actuators, contains the uncertainty of direct AP-actuator channels and cascaded AP-IRS-actuator channels, which can be expressed as
\begin{equation}
\setlength\abovedisplayskip{2pt}
\setlength\belowdisplayskip{2pt}
  \!\!\mathbf{R}_k\!\!=\!\!\left[\mathbf{G}_k^H\quad \mathbf{h}_{d,k}\right]\!\!=\!\![\widehat{\mathbf{G}}_k^H\quad \hat{\mathbf{h}}_{d,k}]\!+\![\Delta\mathbf{G}_k^H\quad\Delta\mathbf{h}_{d,k}]
  \!\!=\!\!\widehat{\mathbf{R}}_k\!+\!\Delta\mathbf{R}_k,
\end{equation}
where
\begin{equation}
\setlength\abovedisplayskip{2pt}
\setlength\belowdisplayskip{2pt}
  \|\Delta\mathbf{R}_k\|_F\!=\!\sqrt{\|\Delta\mathbf{G}_k\|_F^2+\|\Delta\mathbf{h}_{d,k}\|_2^2}
  \!\leq\!\sqrt{\delta_{\mathbf{G}_k}^2+\delta_{\mathbf{h}_{d,k}}^2}\!\triangleq\! \delta_k.
\end{equation}

The robust beamforming problem for imperfect CSI can be given by
\begin{subequations}
\setlength\abovedisplayskip{2pt}
\setlength\belowdisplayskip{2pt}
\begin{align}
\!\!\!\mathrm{ICSI\!-\!P4}:\notag\\
\min_{\mathbf{q,W,V}} \quad&\|\mathbf{q}\|_0\\
\mathrm{s.t.}\quad~ &\mathrm{Tr}(\mathbf{W}\mathbf{R}_k\mathbf{V}\mathbf{R}_k^H)
  \!\!\geq\!\!\sigma_k^2\left(\gamma_{\text{th}}^{(\text{I})}\!-\!q_k\right),\!\forall \|\Delta\mathbf{R}_k\|_F \!\!\leq\!\! \delta_k, \!\forall k,\label{ICSIP4C1}\\
&q_k\geq 0,\forall k,\label{ICSIP4C2}\\
&\|\mathbf{q}\|_0\leq K-1,\\
&\mathrm{Tr}(\mathbf{W})\leq P_{\text{max}},\label{ICSIP4C4}\\
&\mathbf{V}_{mm}=1,m=1,\cdots,M+1,\label{ICSIP4C5}\\
&\mathbf{W}\succeq \mathbf{0},\mathbf{V}\succeq \mathbf{0},\label{ICSIP4C6}\\
&\mathrm{rank}(\mathbf{W})=1,\mathrm{rank}(\mathbf{V})=1.\label{ICSIP4C7}
\end{align}
\end{subequations}
Constraint \eqref{ICSIP4C1} has the following meaning. If the SNR of actuator $k$ can reach the SNR threshold for all possible CSI errors, we consider that the $k$-th actuator can successfully decode the message in the first stage, i.e., $q_k=0$. Otherwise, the decoding of actuator $k$ fails and $q_k$ is a positive number. Furthermore, it is challenging to deal with the infinite number of nonconvex SNR constraints in \eqref{ICSIP4C1} due to the continuity of the channel uncertainty set. The difficulties that remain to be solved are the $\|\mathbf{q}\|_0$ and rank-one constraints.

First, we need to apply the S-procedure in Lemma \ref{S-procedure} to convert the infinite number of SNR constraints in \eqref{ICSIP4C1} into a finite number of constraints. To this end, we recast \eqref{ICSIP4C1} as
\begin{equation}
\setlength\abovedisplayskip{2pt}
\setlength\belowdisplayskip{2pt}
  \!\!\mathrm{vec}^H(\mathbf{R}_k)(\mathbf{V}^T\!\otimes\!\mathbf{W})\mathrm{vec}(\mathbf{R}_k)
  \!\!\geq\!\!\sigma_k^2\left(\gamma_{\text{th}}^{(\text{I})}\!-\!q_k\right),\!\forall \|\Delta\mathbf{R}_k\|_F\!\!\leq \!\! \delta_k,
\end{equation}
by utilizing the following matrix identities $\mathrm{Tr}(\mathbf{ABC})=\mathrm{Tr}(\mathbf{BCA})=\mathrm{Tr}(\mathbf{CAB}),~
\mathrm{Tr}(\mathbf{A}^H\mathbf{B})=\mathrm{vec}^H(\mathbf{A})\mathrm{vec}(\mathbf{B})$, and $\mathrm{vec}(\mathbf{AXB})=(\mathbf{B}^T\otimes\mathbf{A})\mathrm{vec}(\mathbf{X})$.

Further, denote $\mathbf{r}_k=\mathrm{vec}(\mathbf{R}_k),\hat{\mathbf{r}}_k=\mathrm{vec}(\widehat{\mathbf{R}}_k),
\Delta\mathbf{r}_k=\mathrm{vec}(\Delta\mathbf{R}_k)$ bounded by $\|\Delta\mathbf{r}_k\|_2\leq\delta_k$. Thus, constraint \eqref{ICSIP4C1} can be transformed into the following more tractable form as
\begin{align}\label{SNRFullunvertainty}
\setlength\abovedisplayskip{2pt}
\setlength\belowdisplayskip{2pt}
  &\Delta \mathbf{r}_k^H(\mathbf{V}^T\otimes\!\mathbf{W})\Delta \mathbf{r}_k
  +2\mathrm{Re}\{\hat{\mathbf{r}}_k^H(\mathbf{V}^T\otimes\mathbf{W})\Delta \mathbf{r}_k\}
  \notag\\
  &\!+\!\hat{\mathbf{r}}_k^H
  (\mathbf{V}^T\!\otimes\!\mathbf{W})\hat{\mathbf{r}}_k\!\geq\!\sigma_k^2\left(\gamma_{\text{th}}^{(\text{I})}\!-\!q_k\right),
  \forall \|\Delta\mathbf{r}_k\|_2\!\leq\!\delta_k,\forall k.
\end{align}

\begin{lemma}\label{S-procedure}
  (S-procedure for the complex case \cite{Chi2017}): Let $\mathbf{F}_{1}, \mathbf{F}_{2} \in \mathbb{H}^{n}, \mathbf{s}_{1}, \mathbf{s}_{2} \in \mathbb{C}^{n}$, $z_{1}, z_{2} \in \mathbb{R}$. The following implication
\begin{equation}
\setlength\abovedisplayskip{2pt}
\setlength\belowdisplayskip{2pt}
\mathbf{x}^{H} \mathbf{F}_{1} \mathbf{x}\!+\!2\operatorname{Re}\!\left\{\mathbf{s}_{1}^{H} \mathbf{x}\right\}\!+\!z_{1} \!\leq\! 0\!\Rightarrow\! \mathbf{x}^{H} \mathbf{F}_{2} \mathbf{x}\!+\!2\operatorname{Re}\!\left\{\mathbf{s}_{2}^{H} \mathbf{x}\right\}\!+\!z_{2}\!\leq\! 0 \notag
\end{equation}
holds true if and only if there exists a $\mu \geq 0$ such that
\vspace{-1mm}
\begin{equation}
\setlength\abovedisplayskip{4pt}
\setlength\belowdisplayskip{4pt}
\left[\begin{array}{cc}
\mathbf{F}_{2} & \mathbf{s}_{2} \\
\mathbf{s}_{2}^{H} & z_{2}
\end{array}\right] \preceq \mu\left[\begin{array}{cc}
\mathbf{F}_{1} & \mathbf{s}_{1} \\
\mathbf{s}_{1}^{H} & z_{1}
\end{array}\right] \notag
\end{equation}
provided that there exists a point $\hat{\mathbf{x}}$ with $\hat{\mathbf{x}}^{H} \mathbf{F}_{1} \hat{\mathbf{x}}+2 \operatorname{Re}\left\{\mathbf{s}_{1}^{H} \hat{\mathbf{x}}\right\}+z_{1}<0$.
\end{lemma}

According to Lemma \ref{S-procedure}, the infinite number of constraints in \eqref{SNRFullunvertainty} can be represented as the following $K$ matrix inequalities
\begin{align}\label{LMI}
\setlength\abovedisplayskip{2pt}
\setlength\belowdisplayskip{2pt}
&\!\!\!\!\begin{bmatrix}
\mu_k\mathbf{I}_{N_t(M+1)}\!\!+\!\!\mathbf{V}^T\!\!\otimes\!\!\mathbf{W}
&(\!\mathbf{V}^T\!\!\otimes\!\!\mathbf{W})\hat{\mathbf{r}}_k\\
\hat{\mathbf{r}}_k^H(\mathbf{V}^T\!\!\otimes\!\!\mathbf{W}\!)
&\hat{\mathbf{r}}_k^H(\!\mathbf{V}^T\!\!\otimes\!\!\mathbf{W}\!)\hat{\mathbf{r}}_k
\!\!-\!\!\mu_k\delta_k^2\!\!-\!\!\sigma_k^2\left(\!\gamma_{\text{th}}^{(\text{I})}\!\!-\!\!q_k\!\right)
\end{bmatrix}\!\!\succeq\!\!\mathbf{0},\!\forall k,\\
&\Leftrightarrow\pmb{\Xi}_k+\mathbf{E}_k^H(\mathbf{V}^T\otimes\mathbf{W})\mathbf{E}_k\succeq\mathbf{0},
\forall k,\label{sprocedure}
\end{align}
where $\mu_k \geq 0$, $\mathbf{E}_k=[\mathbf{I}_{N_t(M+1)}\quad\hat{\mathbf{r}}_k]$, and
\begin{equation}
\setlength\abovedisplayskip{2pt}
\setlength\belowdisplayskip{2pt}
  \pmb{\Xi}_k=\begin{bmatrix}
\mu_k\mathbf{I}_{N_t(M+1)}
&\mathbf{0}\\
\mathbf{0}
&-\mu_k\delta_k^2-\sigma_k^2\left(\gamma_{\text{th}}^{(\text{I})}-q_k\right)
\end{bmatrix}.
\end{equation}

The reweighted $\ell_1$ method is applied to cope with $\|\mathbf{q}\|_0$. Note that the rank-one constraints in \eqref{ICSIP4C7} are nonconvex and difficult to solve. In view of this, we first drop the rank-one constraints to obtain a relaxed problem. Therefore, we have the following optimization problem
\begin{subequations}
\setlength\abovedisplayskip{2pt}
\setlength\belowdisplayskip{2pt}
\begin{align}
\!\!\!\mathrm{ICSI\!-\!P5}:\min_{\substack{\mathbf{q,W,V}\\\pmb{\mu}\succeq \mathbf{0}}} \quad&\sum_{k=1}^K \omega_kq_k\\
\mathrm{s.t.}\quad~ &\pmb{\Xi}_k+\mathbf{E}_k^H(\mathbf{V}^T\otimes\mathbf{W})\mathbf{E}_k\succeq\mathbf{0},
\forall k,\label{ICSIP5C1}\\
&\sum_{k=1}^K \omega_kq_k\leq K-1,\label{ICSIP5C2}\\
&\eqref{ICSIP4C2},\eqref{ICSIP4C4}-\eqref{ICSIP4C6}.\notag
\end{align}
\end{subequations}

It is essential to emphasize that $\mathbf{W},\mathbf{V}$ remain coupled in the form of the Kronecker product in the matrix inequality constraint \eqref{ICSIP5C1}, which is different from the coupling problem of $\mathbf{W},\mathbf{V}$ in the real-valued inequality constraint \eqref{SNRreformu} as mentioned in Section \ref{SCAsection}. Therefore, we cannot leverage a similar technique for the product decomposition. In light of this, we use a BCD method to split problem $\mathrm{ICSI\!-\!P5}$ into two subproblems, i.e., we alternately optimize $\mathbf{q,W}$ with given $\mathbf{V}$ and $\mathbf{q,V}$ with given $\mathbf{W}$, which can be seen below. Both two subproblems are semidefinite program (SDP) that can be solved by the CVX tool \cite{grant2014cvx}. The overall SDR-based BCD algorithm is summarized in Algorithm \ref{BCD}.
\begin{subequations}
\setlength\abovedisplayskip{2pt}
\setlength\belowdisplayskip{2pt}
\begin{align}
&\mathrm{ICSI\!-\!P6\!-\!1}:\notag\\
\min_{\mathbf{q,W},\pmb{\mu}} \quad&\sum_{k=1}^K \omega_kq_k\\
&\mathbf{W}\succeq \mathbf{0},\pmb{\mu}\succeq \mathbf{0},\\
&\eqref{ICSIP5C1},\eqref{ICSIP5C2},\eqref{ICSIP4C2},\eqref{ICSIP4C4}.\notag
\end{align}
\end{subequations}

\begin{subequations}
\setlength\abovedisplayskip{2pt}
\setlength\belowdisplayskip{2pt}
\begin{align}
&\mathrm{ICSI\!-\!P6\!-\!2}:\notag\\
\min_{\mathbf{q,V},\pmb{\mu}} \quad&\sum_{k=1}^K \omega_kq_k\\
&\mathbf{V}\succeq \mathbf{0},\pmb{\mu}\succeq \mathbf{0},\\
&\eqref{ICSIP5C1},\eqref{ICSIP5C2},\eqref{ICSIP4C2},\eqref{ICSIP4C5}.\notag
\end{align}
\end{subequations}

\begin{algorithm}[t]
\caption{\small \textbf{:} SDR-based BCD Algorithm for Solving Problem $\mathrm{P3}$ with Imperfect CSI}
\label{BCD}
\begin{spacing}{1.2}
	\small{\begin{algorithmic}[1]
		\STATE \textbf{Initialize} $\mathbf{W}^{(0)},\mathbf{V}^{(0)}$, convergence tolerance $\epsilon$, and set outer iteration index $t=0$.
        \STATE Use bisection search method to obtain $\gamma_{\text{th}}^{(\text{I})}$.
		\REPEAT
        \STATE Given $\mathbf{V}^{(t)}$, inner iteration index $i=0$
        \REPEAT
        \STATE Update $\mathbf{q}^{(i+1)},\mathbf{W}^{(i+1)}$ by solving problem $\mathrm{ICSI\!-\!P6\!-\!1}$ with given $\mathbf{W}^{(i)}$.
        \STATE $i\leftarrow i+1$.
        \UNTIL $|\mathcal{A}_1^{(i)}-\mathcal{A}_1^{(i-1)}|\leq \epsilon$.
        \STATE Set $\mathbf{q}^{(t)}=\mathbf{q}^{(i)},\mathbf{W}^{(t+1)}=\mathbf{W}^{(i)}$.
        \STATE Given $\mathbf{W}^{(t+1)}$ , inner iteration index $j=0$.
        \REPEAT
        \STATE Update $\mathbf{q}^{(j+1)},\mathbf{V}^{(j+1)}$ by solving problem $\mathrm{ICSI\!-\!P6\!-\!2}$ with given $\mathbf{V}^{(j)}$.	
        \STATE $j\leftarrow j+1$.	
        \UNTIL $|\mathcal{A}_2^{(j)}-\mathcal{A}_2^{(j-1)}|\leq \epsilon$.
        \STATE Set $\mathbf{q}^{(t+1)}=\mathbf{q}^{(j)},\mathbf{V}^{(t+1)}=\mathbf{V}^{(j)}$.
        \STATE $t\leftarrow t+1$.
        \UNTIL $\|\mathbf{q}^{(t+1)}\|_1-\|\mathbf{q}^{(t}\|_1\leq \epsilon$.
        \STATE
	\end{algorithmic}}
\end{spacing}
\end{algorithm}

The objective function $\mathcal{A}=\sum_{k=1}^K \omega_kq_k$ follows that $\mathcal{A}(\mathbf{W}^{(i)},\mathbf{V}^{(i)})\geq
\mathcal{A}(\mathbf{W}^{(i+1)},\mathbf{V}^{(i)})\geq
\mathcal{A}(\mathbf{W}^{(i+1)},\mathbf{V}^{(i+1)})$, where $\mathbf{W}^{(i+1)},\mathbf{V}^{(i+1)}$ are the optimal solutions obtained for subproblems $\mathrm{ICSI\!-\!P6\!-\!1}$ and $\mathrm{ICSI\!-\!P6\!-\!2}$, respectively. The inequalities hold as $\gamma_k^{(\text{I})}$ can be maximized with optimal $\mathbf{W}^{(i+1)},\mathbf{V}^{(i+1)}$, thereby reducing $q_k$ and lowering the objective function. Therefore, the algorithm can be solved in polynomial time \cite{RazaviyaynSJO2013} with guaranteed convergence.

According to \cite[Theorem 3.12]{ImrePolik2010}, the complexity of an SDP problem with $m$ SDP constraints, where each constraint involves an $n \times n$ PSD matrix, is given by $\mathcal{O}\left(\sqrt{n} \log \frac{1}{\epsilon}\left(m n^{3}+m^{2} n^{2}+m^{3}\right)\right)$, where $\epsilon$ is the convergence tolerance. For SDP problems $\mathrm{ICSI\!-\!P6\!-\!1}$ and $\mathrm{ICSI\!-\!P6\!-\!2}$, we have $m=K+1, n=N_{\mathrm{t}}$ and $m=K+1, n=M+1$, respectively. Thus, the computational complexity of each iteration of the proposed BCD algorithm is $\mathcal{O}(\log \frac{1}{\epsilon}((\sqrt{N_t}+\sqrt{M})K^3+(N_t^{2.5}+M^{2.5}) K^2+(N_t^{3.5}+M^{3.5}) K))$.

When the SDR-based BCD algorithm converges, the obtained solutions $\mathbf{W,V}$ are not always rank-one. To obtain feasible rank-one solutions, we perform the following Gaussian randomization procedure \cite{Wang2020}. First, we decompose $\mathbf{V}$ as $\mathbf{V}=\mathbf{U}_{v} \pmb{\Sigma}_{v} \mathbf{U}_{v}^{H}$, where $\mathbf{U}_{v}=$ $\left[u_{v, 1}, \ldots, u_{v, M+1}\right]$ and $\pmb{\Sigma}_{v}=\operatorname{diag}\left(\lambda_{v, 1}, \ldots, \lambda_{v, M+1}\right)$ are unitary matrices containing eigenvectors and a diagonal matrix with eigenvalues on its diagonal, respectively. Then, we generate a vector satisfying $\mathbf{v}=\mathbf{U}_{v} \pmb{\Sigma}_{v}^{\frac{1}{2}} \mathbf{e}_v$, where $\mathbf{e}_v \sim \mathcal{C} \mathcal{N}\left(\mathbf{0}, \mathbf{I}_{M+1}\right)$. Further, we normalize each element of $\mathbf{v}$ and obtain $\mathbf{v}^{*}$ as the optimal phase shift. Next, we decompose $\mathbf{W}$ as $\mathbf{W}=\mathbf{U}_w^{H} \pmb{\Sigma}_w \mathbf{U}_w$ and obtain a sub-optimal solution as $\mathbf{w}=\mathbf{U}_w^{H} \pmb{\Sigma}_w^{\frac{1}{2}} \mathbf{e}_w$, where $\mathbf{e}_w \sim \mathcal{C} \mathcal{N}\left(\mathbf{0}, \mathbf{I}_{N_t}\right)$. For randomly generated vectors $\mathbf{e}_w$, the best $\mathbf{w}^{*}$ is selected if it satisfies the transmit power constraint \eqref{ICSIP5C1} and the SNR constraint \eqref{ICSIP4C4} while minimizing the objective function.
\section{Simulation Results}\label{simuresults}
In this section, we provide simulation results to evaluate the performance of the proposed two-stage IRS-aided D2D communication protocol in terms of reliable communication capability between AP and actuators under the perfect and imperfect CSI scenarios.
\vspace{-3mm}
\subsection{Simulation Setup}
The AP and IRS are located at (0,0,15) meter (m) and (50,50,15) m, respectively, and all actuators are randomly and uniformly distributed on a circle centered at (50,50,0) m with a radius of 20 m. The path loss is modeled as $\mathrm{PL}=C_0d^{-\alpha}$ \cite{MaoIWCL2021}, where $C_0$, set as -30 dB, is the path loss at the reference distance 1 m, $d$ is the link distance, and $\alpha$ is the path loss exponent. The small-scale fading of AP-IRS and IRS-actuator channels are assumed to be Ricean fading. The direct link between AP and actuators and the D2D link are modeled as Rayleigh fading. Assume the system operates in a mini-slot fashion, i.e., $\tau_1=\tau_2$. The carrier frequency is set as 2.4 GHz. Considering the imperfect CSI effects, an uncertainty ratio, denoted by $\kappa$, is defined as $\kappa=\frac{\delta_k}{\|\mathbf{R}_k\|_F}$, where $\kappa\in[0,1]$ \cite{LiIToSP2011}. Unless otherwise specified, other parameters are listed in Table \ref{table1}.
\begin{table}
\setlength{\belowcaptionskip}{-3mm}
\begin{center}
\caption{Simulation Parameters}
\label{table1}
\small{\begin{tabular}{ccc}
\toprule
Symbol&Definition&Typical value\\
\midrule
$\alpha_1$&\makecell[c]{Path loss exponent for\\ Ricean fading channels} & 2 \\
$\alpha_2$&\makecell[c]{Path loss exponent for\\ Rayleigh fading channels}& 4\\
$\beta$&Ricean factor&2\\
$\sigma_k^2,\forall k$&Noise power  & -70 dBm \\
$B$&Bandwidth & 0.5 MHz \\
$P_{\max}$ & Transmit power at AP  & 43 dBm \\
$P$ & Transmit power at actuators & 23 dBm \\
$\varepsilon_{\text{th}}$& Maximum PEP & $10^{-6}$ \\
$\epsilon$& Convergence tolerance & $10^{-4}$ \\
\bottomrule
\end{tabular}}
\end{center}
\end{table}

\begin{figure}
\centering
\includegraphics[width=0.95\linewidth]{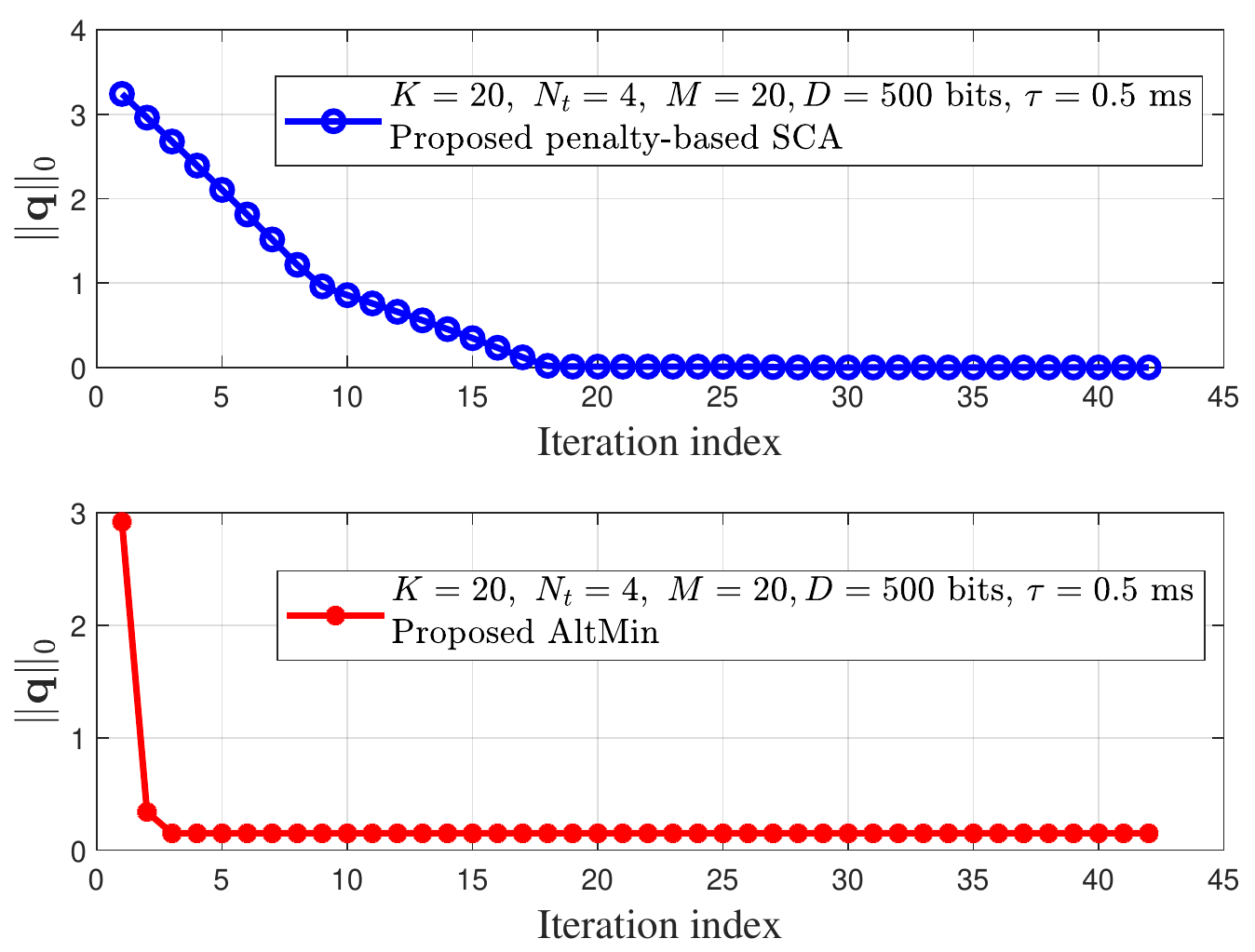}
\caption{Convergence of the proposed algorithms.}
\label{convergence}
\end{figure}

To evaluate the proposed protocol in terms of reliability, we define a performance metric called Probability of Reliable Communication (PRC) as the ratio of the number of experiments where all actuators can successfully decode the message over total experiments. For comparison, we consider the following schemes: 1) \textbf{Upper bound}: we relax the rank-one constraints, i.e., $\mathrm{rank}(\mathbf{W})=1,\mathrm{rank}(\mathbf{V})=1$ and then solve problem $\mathrm{P3}$ until convergence, which serves as a performance upper bound; 2) \textbf{Proposed AltMin algorithm}: the approach is proposed in Section \ref{AltMinSection} for the perfect CSI scenario; 3) \textbf{Proposed penalty-based SCA algorithm}: the method is proposed in Section \ref{SCAsection} for the perfect CSI case; 4) \textbf{Baseline 1}: IRS is deployed for packet transmission, but we adopt random phase shifts; 5) \textbf{Baseline 2}: IRS is not introduced into the system; 6) \textbf{Proposed SDR-based BCD algorithm}: the approach is proposed in Section \ref{worstcasesec} for the imperfect CSI scenario; 7) \textbf{Penalty-based BCD algorithm}: the penalty-based method is used to handle the rank-one constraints for the imperfect CSI case.
\vspace{-3mm}
\subsection{Convergence of the Proposed Algorithms}
We plot the convergence performance of the proposed AltMin algorithm and the proposed penalty-based SCA algorithm for $K=20,N_t=4,M=20,D=500$ bits, $\tau=0.5$ ms. As it can be seen from Fig. \ref{convergence}, both algorithms converge monotonically, and the proposed penalty-based SCA algorithm requires significantly more iterations to converge than the proposed AltMin algorithm. In particular, the proposed AltMin algorithm converges after about 5 iterations on average, while the penalty-based SCA algorithm needs another 30 iterations on average to converge. This is because the number of optimization variables and
constraints in the penalty-based SCA algorithm is much higher than those in the AltMin algorithm, which was discussed in more detail in Section \ref{AltMinSection} and Section \ref{SCAsection}.
\vspace{-3mm}
\subsection{Impact of Packet Size}
Fig. \ref{PRCvsbits} shows the PRC performance versus the packet size $D$ for different schemes when $K=20,N_t=2,M=10$, $\tau=1$ ms. It can be seen that the proposed AltMin algorithm and the penalty-based SCA algorithm can achieve reliable communication within 700 data bits. Moreover, AP and actuators in both proposed algorithms can also communicate reliably at 800 bits with $99\%$ probability. Particularly, the PRC performance of the two proposed algorithms is close to the upper bound scheme, which demonstrates the effectiveness of the proposed algorithms. In contrast, with the random-phase IRS, reliable communication between AP and actuators can be achieved within 500 bits in the baseline 1 scheme. The PRC performance degrades significantly compared to the proposed algorithms with optimal reflective beamforming. This shows the importance of optimizing the reflection matrix, which can greatly improve the quality of the received signal by intelligent reflection. In the baseline 2 scheme, without the support of IRS, the actuators can only communicate reliably within 300 bits, which has worse PRC performance. This is because the received SNR in the first stage is significantly reduced without the help of IRS. As a result, the number of actuators with successful decoding in the first stage decreases significantly, i.e., the number of actuators to relay messages in the second stage decreases. Consequently, the number of actuators that can successfully decode the message in the second stage also decreases. From the comparison of the proposed algorithms with baseline 1 and 2, we know that the integration of IRS with optimized phase shifts into the system play an important role in improving reliability, which confirms the necessity of the IRS in the two-stage protocol.
\begin{figure}[t]
\centering
\includegraphics[width=0.95\linewidth]{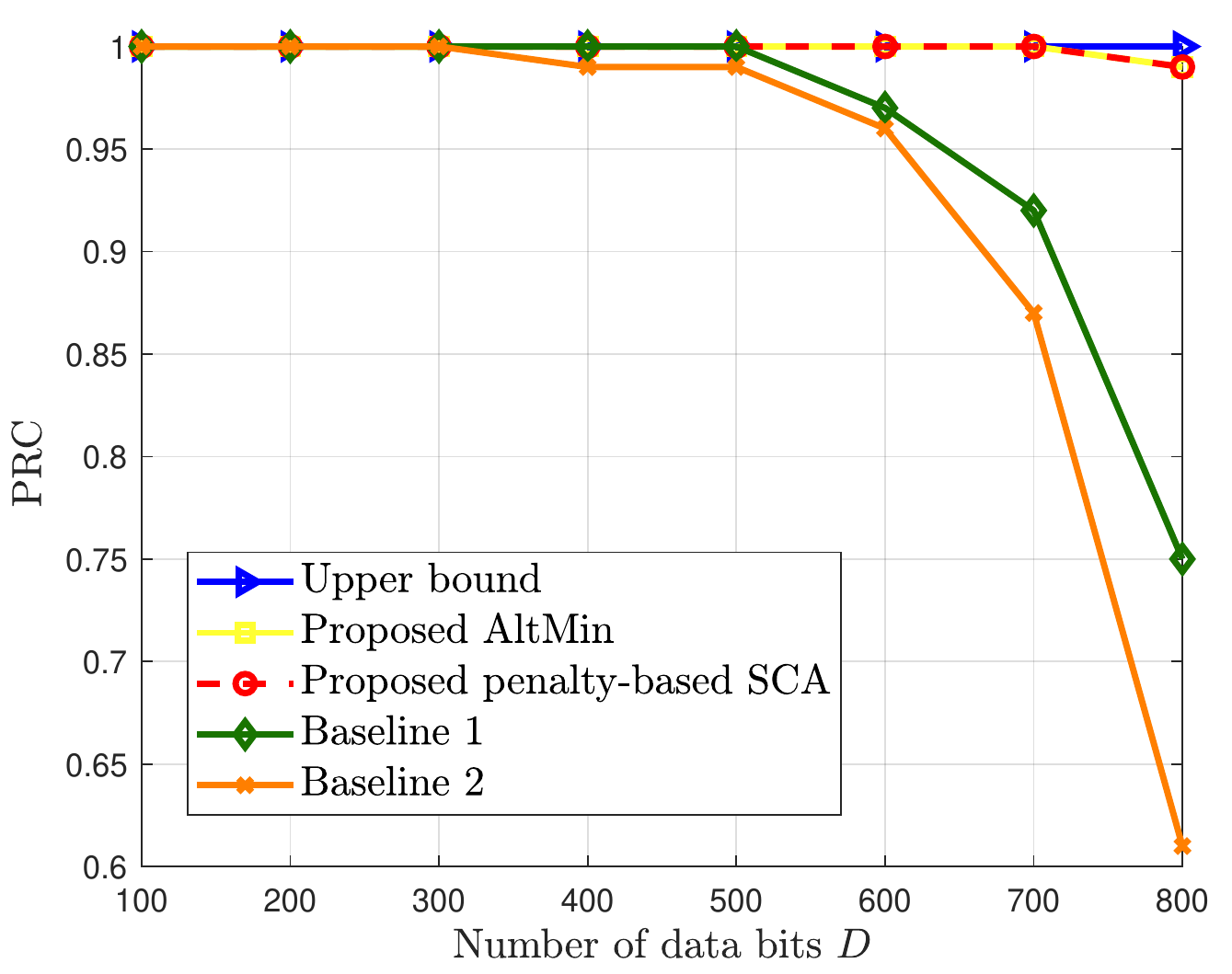}
\caption{PRC performance versus packet size $D$ for different schemes.}
\label{PRCvsbits}
\end{figure}

\begin{figure}[t]
\centering
\includegraphics[width=0.95\linewidth]{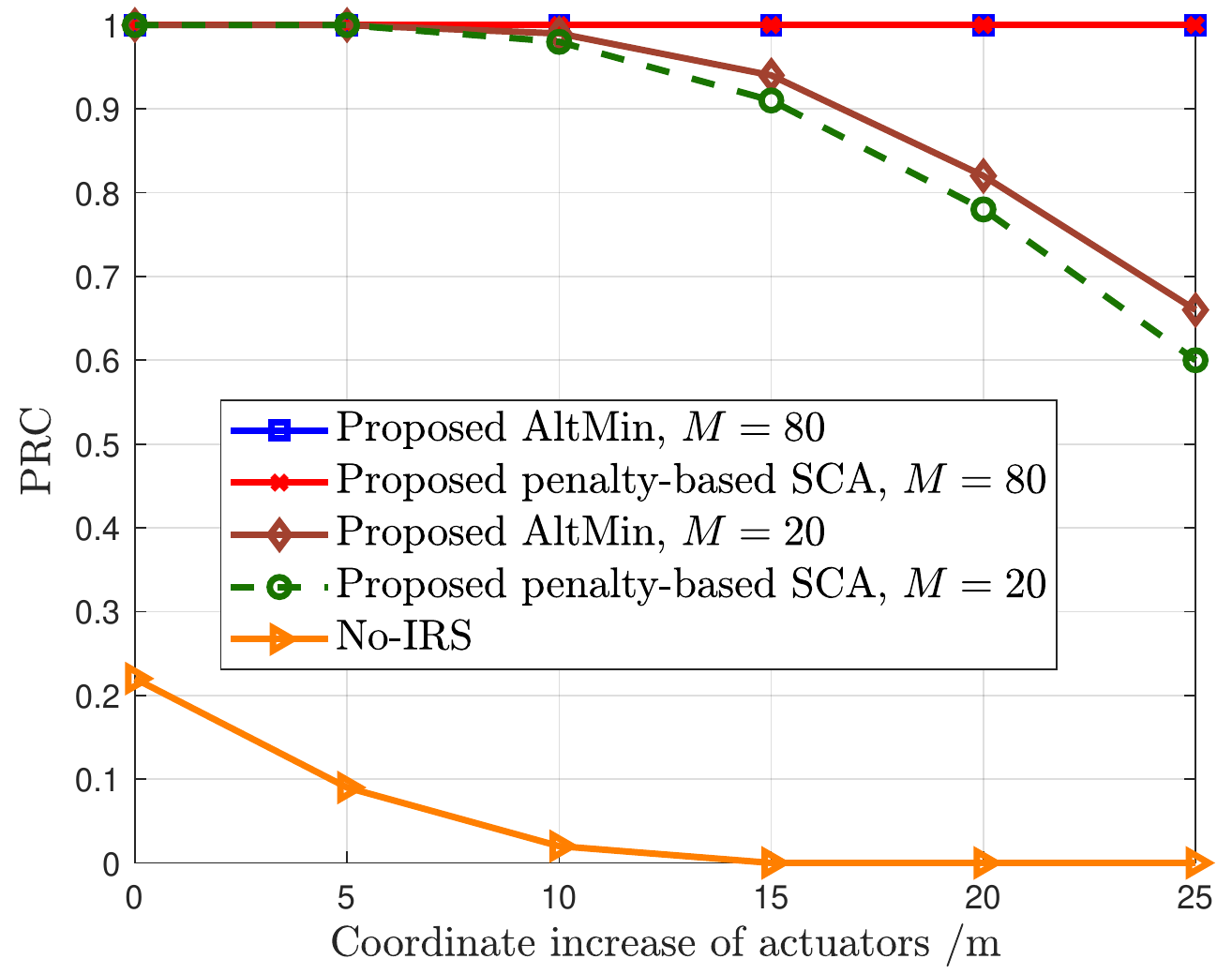}
\caption{Necessity of IRS on the PRC performance for different location of actuators.}
\label{PRCvslocation}
\end{figure}
\vspace{-3mm}
\subsection{Necessity of IRS}
To clarify the necessity of IRS in a two-stage communication protocol, we investigate the effect of distance variation between the AP and actuators on the PRC performance by comparing the proposed algorithms and the No-IRS scheme for $K=10,N_t=2,D=900$ bits, $\tau=1$ ms. We set the location of actuators mentioned above as reference points, i.e., all actuators are randomly and uniformly distributed on a circle centered at (50,50,0) m with a radius of 20 m. Then in order to represent the distance variation between the AP and actuators, we increase the horizontal and vertical coordinates of the actuators at the same time. From Fig. \ref{PRCvslocation}, we can observe that when the distance between the AP and actuators is large, the actuators cannot communicate reliably with the AP under the No-IRS scheme. This is because of the poor channel quality of the direct link between the AP and the actuators. After the deployment of IRS into the communication protocol, by increasing the number of IRS reflecting elements, actuators can successfully receive the messages sent by the AP even when they are far from each other.
\vspace{-3mm}
\subsection{Impact of Second-stage D2D Transmission}
We now examine the role of the second-stage D2D transmission. For this purpose, we plot the average number of actuators with successful decoding in each stage and the total number in two stages under the proposed AltMin algorithm and the baseline 2 scheme for different packet size D. From Fig. \ref{D2DROLE}, it can be seen that the average number of actuators with successful decoding in the first stage decreases as the packet size $D$ increases. Accordingly, the remaining actuators with decoding failures must rely on the second-stage D2D transmission to achieve reliable communication. In particular, the D2D network begins to function at $D>200$
bits in the proposed AltMin algorithm and at $D>100$ bits in the baseline 2 (i.e., No-IRS). Moreover, when the size of the data packets is larger, the second-stage D2D transmission plays a greater role. In addition, we find that when the data packet becomes larger (taking 800 bits as an example), about $30\%$ of the actuators in the proposed algorithm rely on the D2D network to achieve reliable communication, while in the baseline 1 scheme, the number of users who need to rely on D2D network to realize reliable communication is up to about $70\%$. This shows that the IRS-assisted first-stage transmission can effectively reduce the communication load of the second-stage D2D network.
\begin{figure}[!t]\centering
	\subfloat[Proposed AltMin algorithm.]{
    \label{RISD2D}
    \includegraphics[width=.95\linewidth]{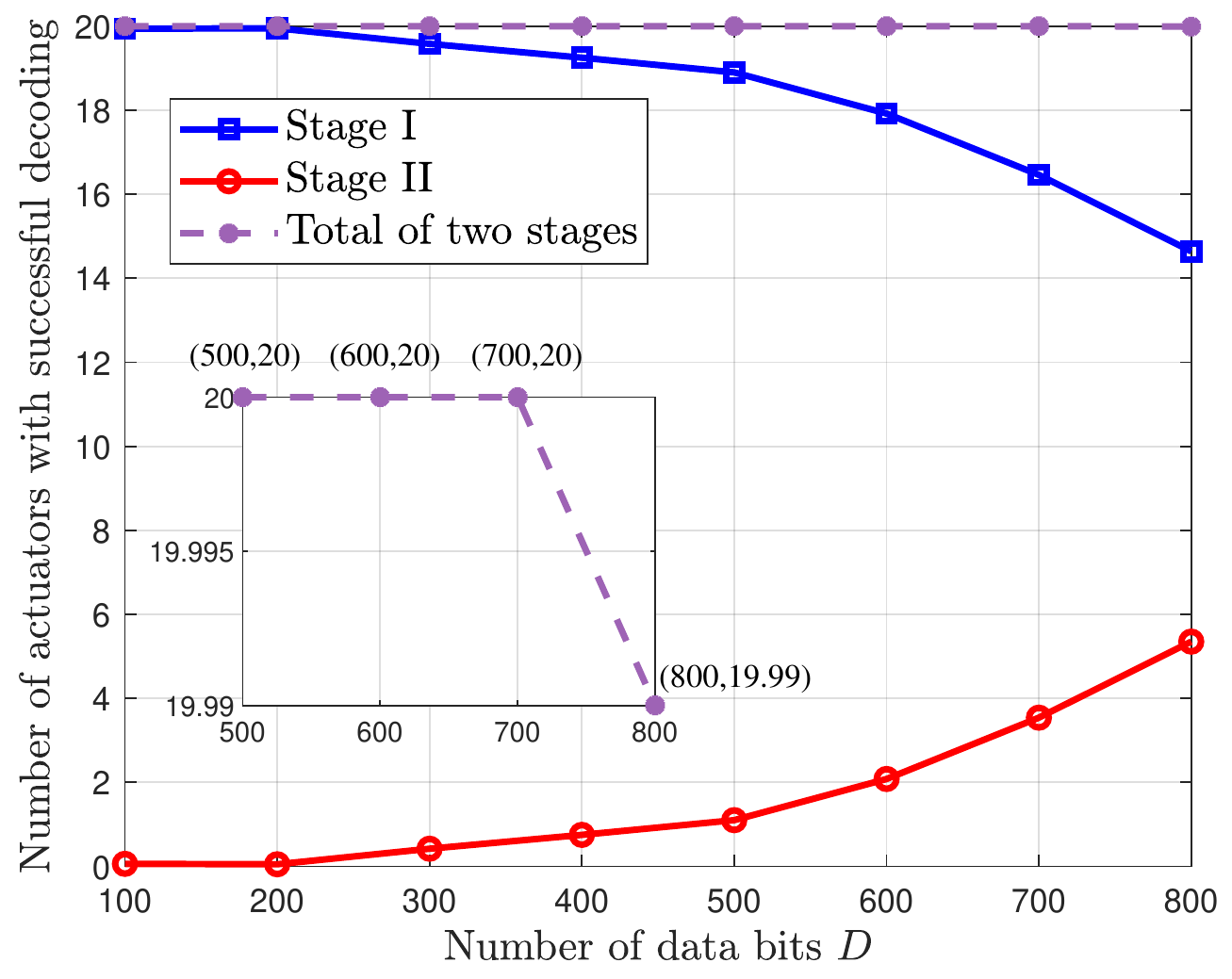}}\\[0.01mm]
	\subfloat[Baseline 2.]{
    \label{NORISD2D}
    \includegraphics[width=.95\linewidth]{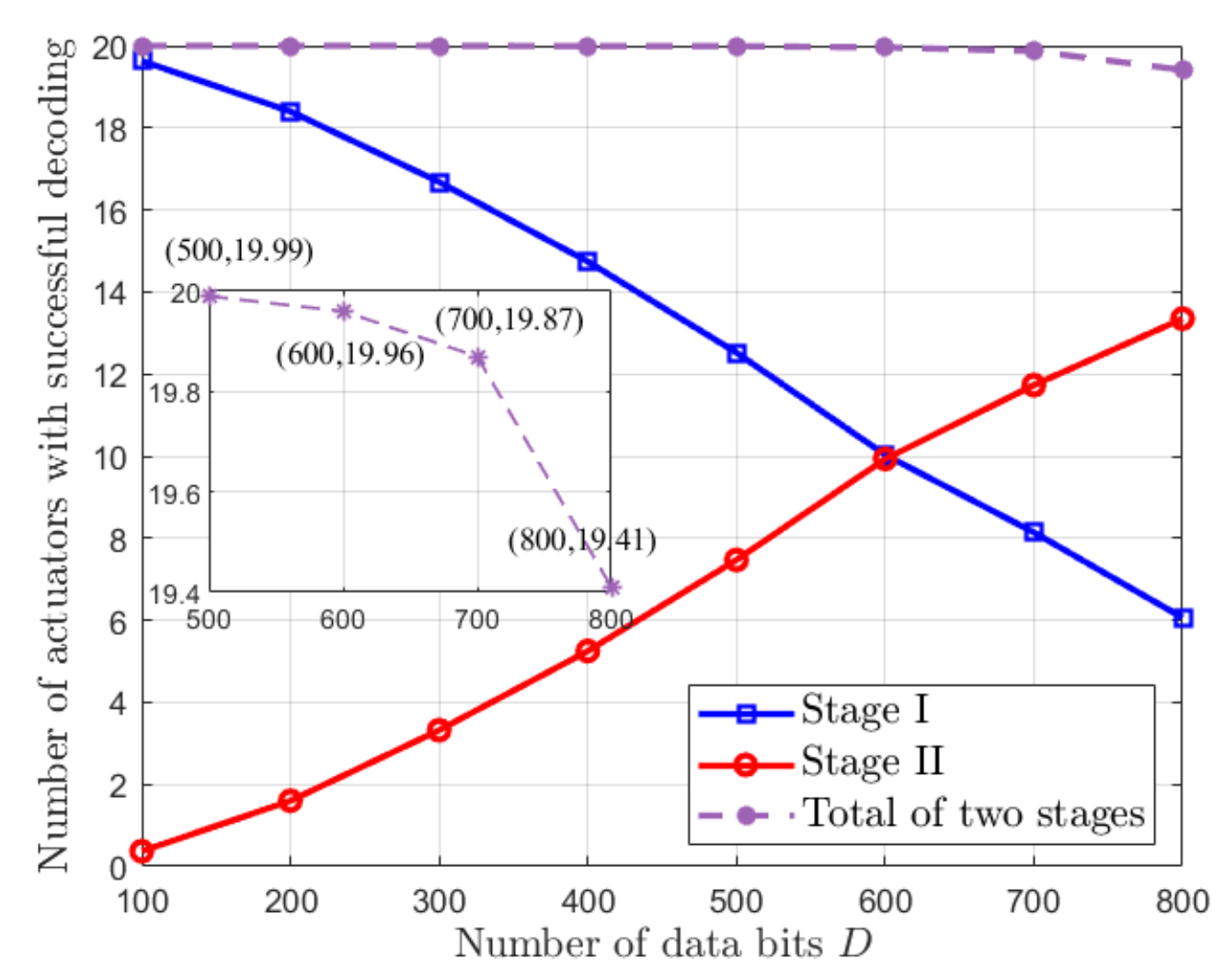}}
	\caption{Average number of actuators with successful decoding in each stage and the total number in two stages under the proposed AltMin algorithm and baseline 2 scheme for different $D$.}
	\label{D2DROLE}
\end{figure}
\vspace{-3mm}
\subsection{Impact of Delay}
In Fig. \ref{PRCvstau}, we study the effects of delay on the PRC performance for different schemes when $K=20,N_t=4,M=20,D=500$ bits. It can be observed that the baseline 1 and 2 schemes are sensitive to delay and their PRC performance is poor, especially for strict latency requirements. The proposed algorithms, on the other hand, can ensure reliable communication even for small delay (e.g. $\tau=0.5$ ms). This is due to the doubly improved reliability of the proposed two-stage protocol through the combined use of IRS and the D2D network. Thus, the proposed two-stage IRS-aided D2D communication protocol is crucial and better copes with URLLC-oriented applications.
\vspace{-3mm}
\subsection{Impact of Number of Reflecting Elements}\label{impactM}
Fig. \ref{PRCvsM} depicts the PRC performance versus the number of reflecting elements $M$ for different $K$ and schemes when $N_t=2,D=900$ bits, $\tau=1$ ms. It can be seen that the proposed algorithms perform almost the same and their PRC performance approaches the upper bound. When the number of reflecting elements $M$ increases, the PRC performance improves and reliable communication of all actuators can be guaranteed. Although IRS is theoretically passive, i.e., it does not actively send and receive signals, it also requires a power supply to maintain the operation of each reflecting element and the intelligent controller. So in this case, 20 reflecting elements are sufficient for all actuators to successfully decode the 900-bit packet within 1 ms, which can provide some insight into the practical application. Furthermore, compared to baseline 1, the proposed algorithms have a large performance gap, which becomes smaller as the number of reflecting elements increases. This is because more reflecting elements can provide higher spatial degree of freedom (DoF), significantly improving the received signal. In particular, the PRC performance for $K=20$ is better than that for $K=10$, which benefits from the multiuser diversity and the second-stage D2D transmission between actuators with proximity to each other. At the same time, it is worth noting that the message that AP sends in the first stage and the successful-decoding actuators relay in the second stage is the same combined message, so there is no interference between the actuators. More specifically, the number of
actuators that can successfully decode the signal in the first stage is higher for $K=20$ than for $K=10$, i.e., the number of actuators that act as relays to relay messages in the second stage becomes larger, increasing the probability of successful reception of the remaining actuators.

\begin{figure}[t]
\centering
\includegraphics[width=0.95\linewidth]{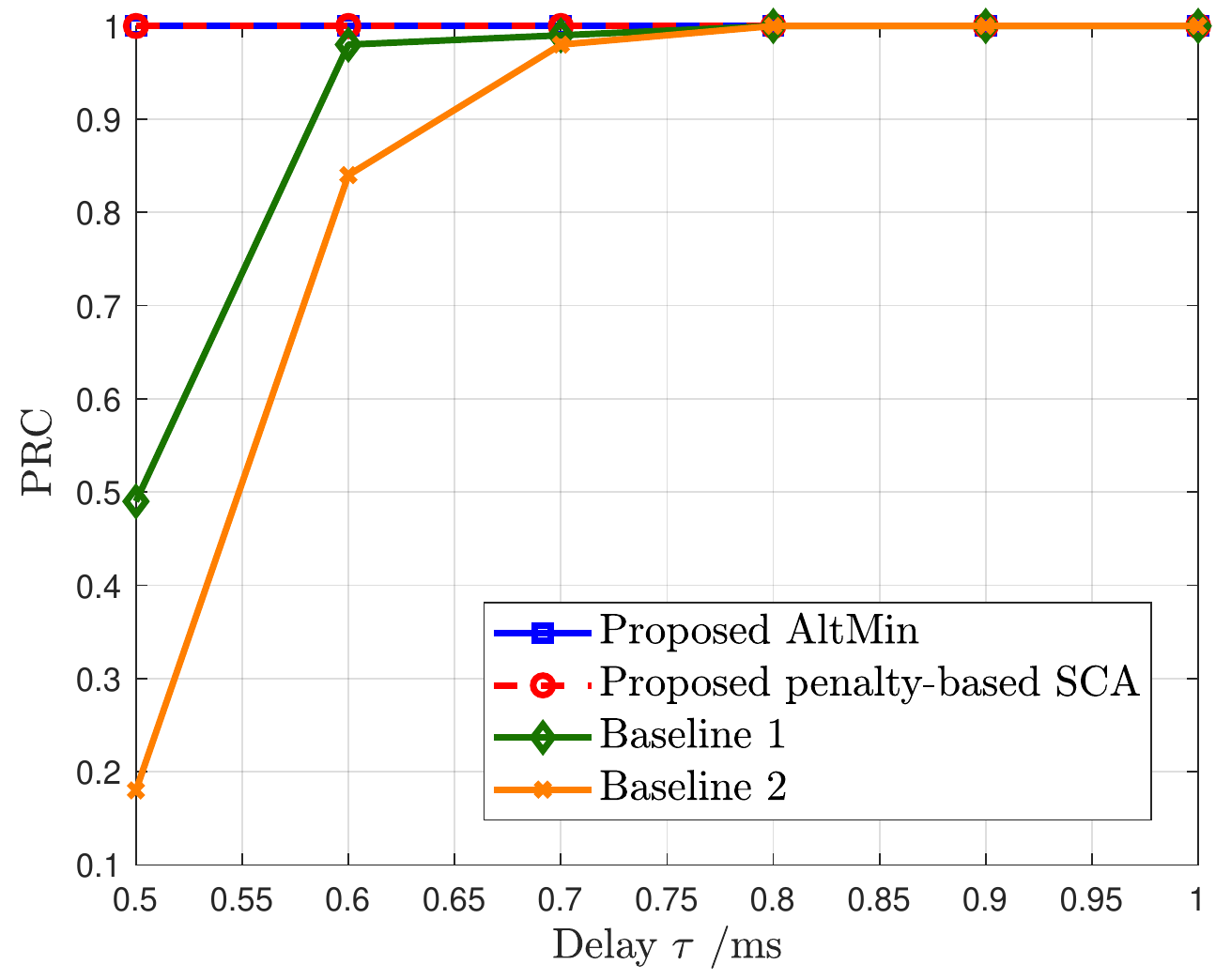}
\caption{PRC performance versus delay $\tau$ for different schemes.}
\label{PRCvstau}
\end{figure}

\begin{figure}[t]
\centering
\includegraphics[width=0.95\linewidth]{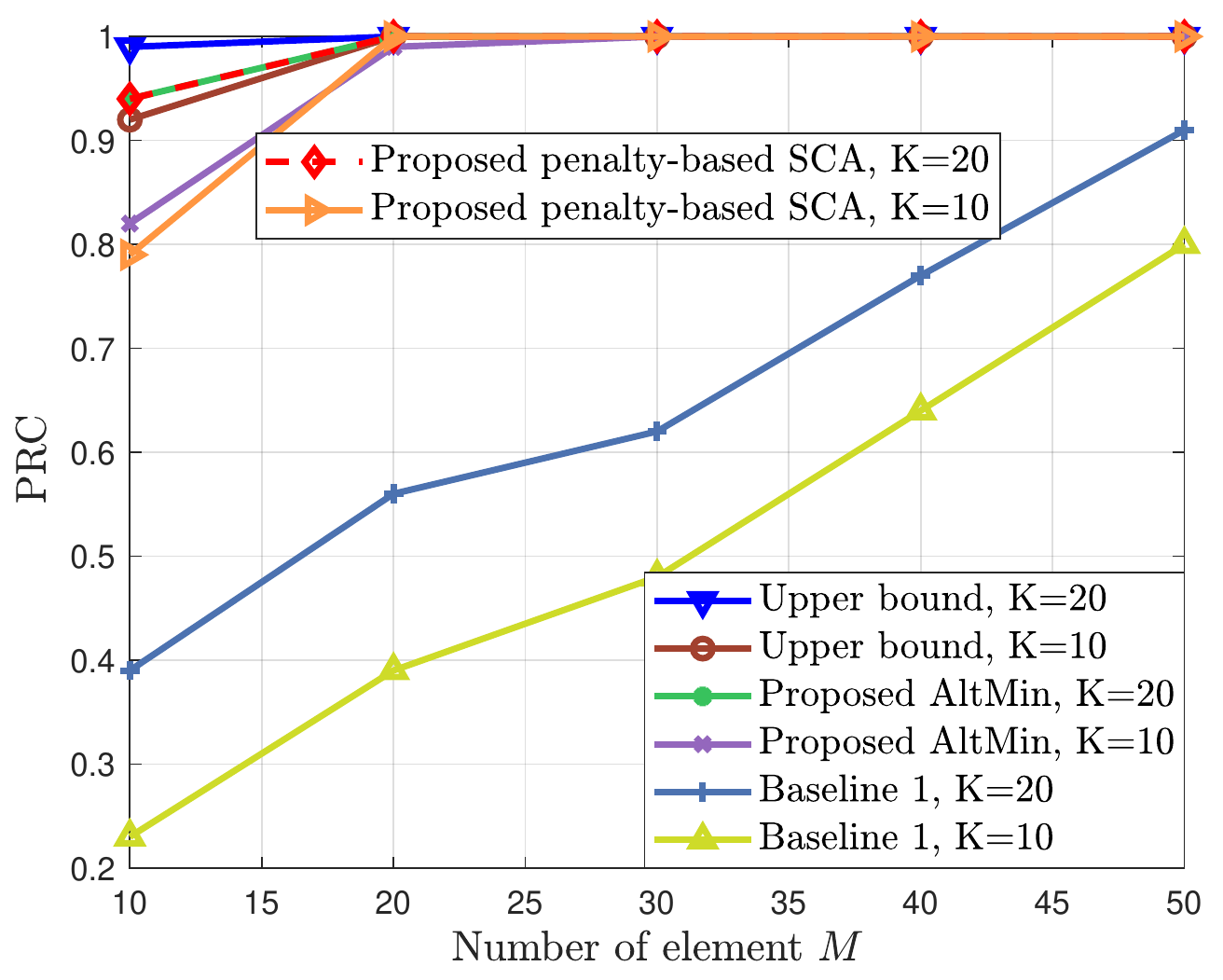}
\caption{PRC performance versus number of reflecting elements $M$ for different $K$ and schemes.}
\label{PRCvsM}
\end{figure}
\vspace{-3mm}
\subsection{Impact of Number of Antennas}
Fig. \ref{PRCvsNt} shows the effects of the number of antennas $N_t$ on the PRC performance for different $K$ and schemes when $M=10,D=900$ bits, $\tau=1$ ms. As it can be seen from Fig. \ref{PRCvsNt}, the probability of reliable communication between AP and actuators improves as the number of antennas increases since more antennas provide a higher diversity gain. The PRC performance of the proposed AltMin algorithm is slightly better than that of the penalty-based SCA algorithm, and the performance gap between the proposed algorithms and the upper bound narrows as $N_t$ increases. The reason why the PRC performance for $K=20$ is better than that for $K=10$ was discussed in Section \ref{impactM}. In addition, when we compare the proposed algorithms with baseline 1, we also find that as the number of antennas increases, the received signal of the direct link becomes stronger and the influence of the reflection matrix optimization becomes weaker. That is, reliable communication via the random phase IRS can be guaranteed if the number of antennas is sufficiently large, but it is achieved at the expense of the high cost of maintaining the antenna arrays. This demonstrates the necessity of using IRS with well-optimized phase shifts.
\vspace{-3mm}
\subsection{Impact of CSI Uncertainty}
In Fig. \ref{PRCvskappa}, we investigate the impact of CSI uncertainty on the PRC performance for different schemes when $K=10,N_t=4,M=10,D=400$ bits, $\tau=1$ ms. It can be seen that the PRC performance deteriorates when the accuracy of the CSI estimate decreases. This is because the difficulty of performing accurate active beamforming at AP and reflective beamforming at IRS increases with the poorer quality of the CSI estimate, resulting in poorer PRC performance. The proposed SDR-based BCD algorithm outperforms the penalty-based BCD algorithm and baseline 1 (i.e., IRS with random phase shifts) over the entire range of the considered CSI uncertainty ratio, which shows that the proposed algorithm can fully and efficiently exploit the spatial DoF to improve the reliability even when the CSI uncertainty exists.The proposed algorithm can guarantee reliable communication with a probability of $98\%$, even in the presence of large CSI estimation errors. This confirms the robustness of the proposed algorithm to CSI uncertainty.
\begin{figure}[t]
\centering
\includegraphics[width=0.95\linewidth]{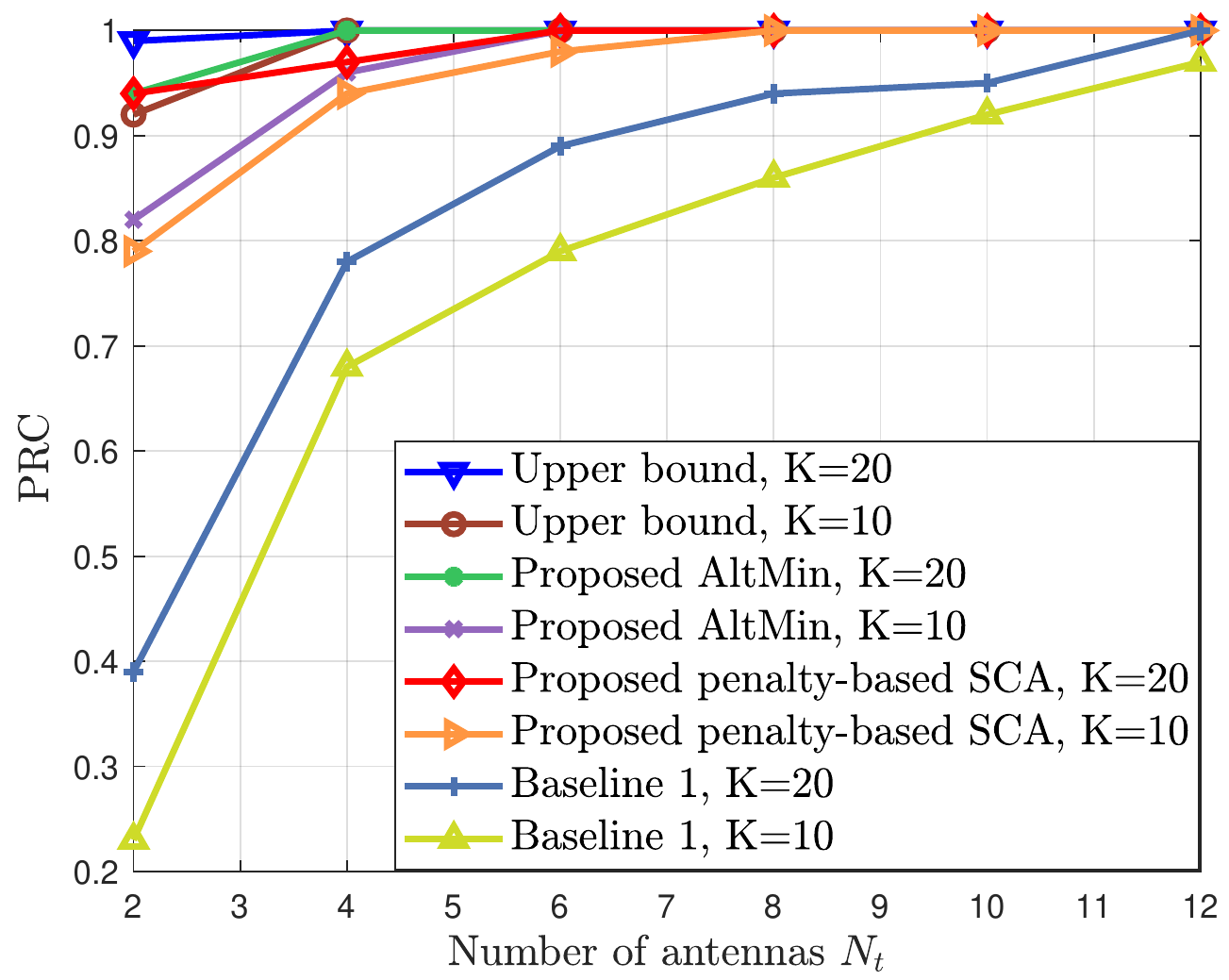}
\caption{PRC performance versus number of antennas $N_t$ for different $K$ and schemes.}
\label{PRCvsNt}
\end{figure}

\begin{figure}[t]
\centering
\includegraphics[width=0.95\linewidth]{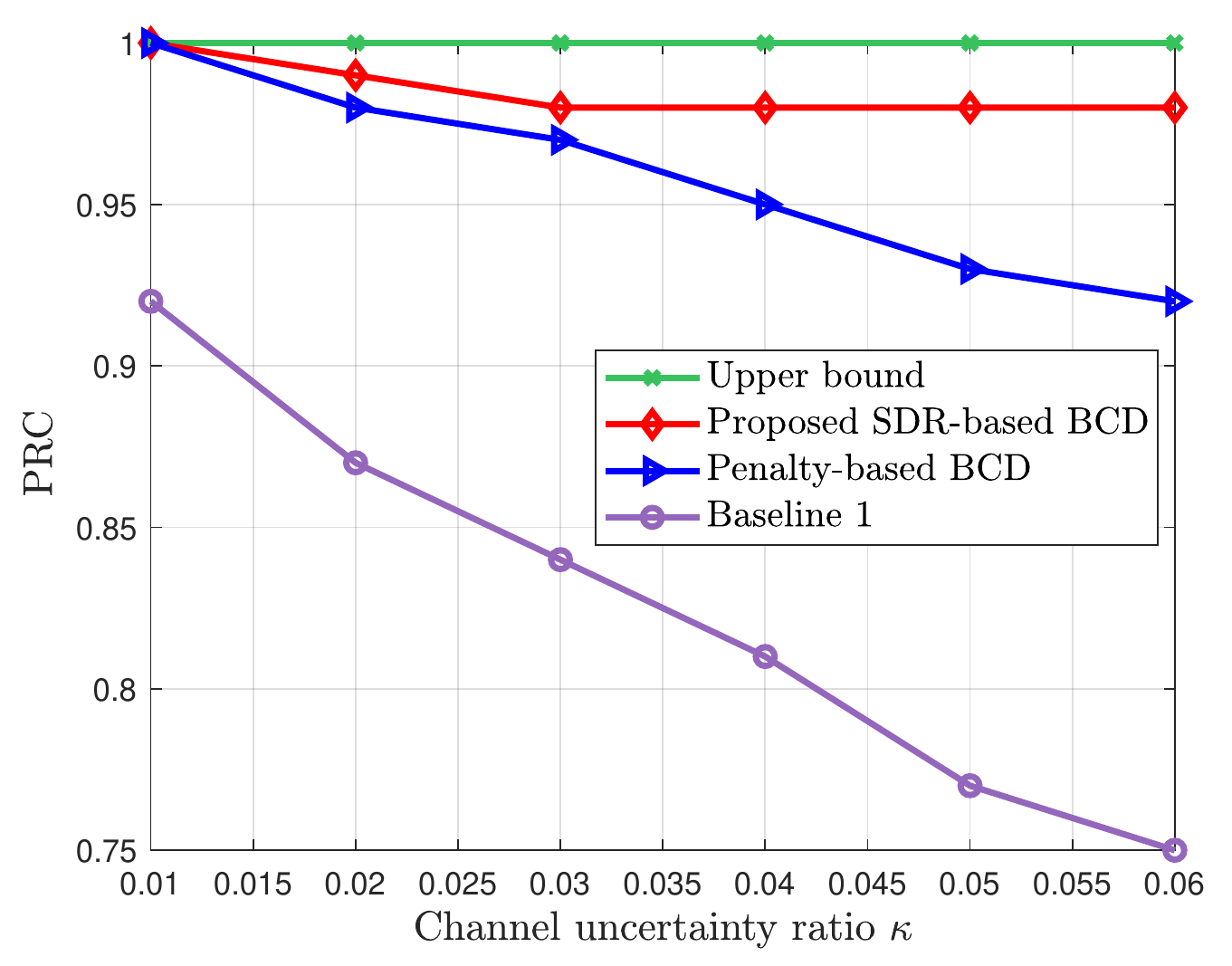}
\caption{PRC performance versus channel uncertainty ratio $\kappa$ for different schemes.}
\label{PRCvskappa}
\end{figure}
\section{Conclusion}\label{conclusionsec}
In this work, we exploited the potential of IRS and D2D communication to enable URLLC between an AP and multiple actuators in the IIoT scenario. We proposed a two-stage protocol where MRC is adopted for joint decoding the superimposed signal (direct signal from AP and reflected signal from IRS) in the first stage and the signal from D2D links in the second stage. The optimization problem was formulated to maximize the number of actuators with successful reception by jointly optimizing the active beamforming at AP and the phase shifts at IRS. We studied the joint beamforming problem under the scenarios of perfect and imperfect CSI of both direct AP-actuator channels and cascaded AP-IRS-actuator channels, where efficient algorithms with complexity and convergence analysis were proposed for each case. Simulation results confirmed the role and importance of IRS and D2D communication in enhancing reliability compared to other baseline schemes. Thanks to the doubly enhanced reliability via IRS and D2D network, the proposed two-stage protocol can achieve reliable communication even under stringent latency requirements and even in the presence of CSI uncertainties.
\vspace{-2mm}
\begin{appendices}
  \section{Proof of Proposition \ref{Proposition1}}\label{appendixA}
  For problem $\mathrm{SCA\!-\!P4}$, we denote the objective function by $f(\mathbf{W,V})$ and let $\mathbf{W}^{\star},\mathbf{V}^{\star}$ be its optimal solutions. Hence, we have $f\left(\mathbf{W}^{\star},\mathbf{V}^{\star}\right) \leq f(\mathbf{W,V})$. Denote the objective function of problem $\mathrm{SCA\!-\!P5}$ by $g(\mathbf{W,V};\rho)$ and its optimal solutions by $\mathbf{W}_{s},\mathbf{V}_{s}$ for penalty factor $\rho_{s}$. Thus, we have $g\left(\mathbf{W}_{s},\mathbf{V}_{s};\rho_{s}\right) \leq g\left(\mathbf{W}^{\star},\mathbf{V}^{\star} ; \rho_{s}\right)$, i.e.,
\begin{align}\label{C1}
\setlength\abovedisplayskip{4pt}
\setlength\belowdisplayskip{4pt}
  &f\left(\mathbf{W}_{s},\mathbf{V}_{s}\right)+\frac{1}{\rho_{s}}\left(\left\|\mathbf{W}_{s}\right\|_{*}
  -\left\|\mathbf{W}_{s}\right\|_{2}+\left\|\mathbf{V}_{s}\right\|_{*}-\left\|\mathbf{V}_{s}\right\|_{2}\right)
  \notag\\
  &\leq f\left(\mathbf{W}^{\star},\mathbf{V}^{\star}\right)\!+\!\frac{1}{\rho_{s}}\left(
  \left\|\mathbf{W}^{\star}\right\|_{*}\!-\!\left\|\mathbf{W}^{\star}\right\|_{2}\!+\!
  \left\|\mathbf{V}^{\star}\right\|_{*}\!-\!\left\|\mathbf{V}^{\star}\right\|_{2}\right)\notag\\
  &\overset{(a)}{=}f\left(\mathbf{W}^{\star},\mathbf{V}^{\star}\right),
\end{align}
where $(a)$ holds as $\mathbf{W}^{\star},\mathbf{V}^{\star}$ are the optimal solutions of problem $\mathrm{SCA\!-\!P4}$, which satisfy the rank-one constraints, i.e., $\left\|\mathbf{W}^{\star}\right\|_{*}-\left\|\mathbf{W}^{\star}\right\|_{2}=0,
\left\|\mathbf{V}^{\star}\right\|_{*}-\left\|\mathbf{V}^{\star}\right\|_{2}$. From \eqref{C1}, we have
\begin{align}\label{C2}
\setlength\abovedisplayskip{4pt}
\setlength\belowdisplayskip{4pt}
  &\left\|\mathbf{W}_{s}\right\|_{*}
  -\left\|\mathbf{W}_{s}\right\|_{2}+\left\|\mathbf{V}_{s}\right\|_{*}-\left\|\mathbf{V}_{s}\right\|_{2} \notag\\
  &\leq\rho_{s}\left[f\left(\mathbf{W}^{\star},\mathbf{V}^{\star}\right)-
  f\left(\mathbf{W}_{s}\mathbf{V}_{s}\right)\right].
\end{align}

Denote $\overline{\mathbf{W}},\overline{\mathbf{V}}$ by limit points of sequence $\left\{\mathbf{W}_{s},\mathbf{V}_{s}\right\}$ and there is an infinite subsequence $\mathcal{S}$ such that $\lim\limits_{s \in \mathcal{S}} \mathbf{W}_{s}=\overline{\mathbf{W}},\lim\limits_{s \in \mathcal{S}} \mathbf{V}_{s}=\overline{\mathbf{V}}$. By taking the limit as $s \rightarrow \infty$ for $ s \in \mathcal{S}$ on both sides of \eqref{C2}, we can obtain that
\begin{equation}\label{C3}
\setlength\abovedisplayskip{2pt}
\setlength\belowdisplayskip{2pt}
\begin{aligned}
&\lim_{s \in \mathcal{S}}\left(\left\|\mathbf{W}_{s}\right\|_{*}-\left\|\mathbf{W}_{s}\right\|_{2}
+\left\|\mathbf{V}_{s}\right\|_{*}-\left\|\mathbf{V}_{s}\right\|_{2}\right) \notag\\ &\overset{(b)}{=}\|\overline{\mathbf{W}}\|_{*}-\|\overline{\mathbf{W}}\|_{2}+\|\overline{\mathbf{V}}\|_{*}
-\|\overline{\mathbf{V}}\|_{2}  \\
&\leq \lim _{s \in \mathcal{S}} \rho_{s}\left[f\left(\mathbf{W}^{\star},\mathbf{V}^{\star}\right)-f\left(\mathbf{W}_{s},\mathbf{V}_{s}\right)\right] \overset{\rho_{s} \rightarrow 0}{=} 0,
\end{aligned}
\end{equation}
where $(b)$ holds owing to the continuity of the function $\|\mathbf{X}\|_{*}-\|\mathbf{X}\|_{2}$. Combing \eqref{C3} with $\|\overline{\mathbf{W}}\|_{*}-\|\overline{\mathbf{W}}\|_{2}\geq 0, \|\overline{\mathbf{V}}\|_{*}
-\|\overline{\mathbf{V}}\|_{2}\geq 0$, we have $\|\overline{\mathbf{W}}\|_{*}-\|\overline{\mathbf{W}}\|_{2}=0,\|\overline{\mathbf{V}}\|_{*}
-\|\overline{\mathbf{V}}\|_{2}=0$, so $\overline{\mathbf{W}},\overline{\mathbf{V}}$ are feasible for problem $\mathrm{SCA\!-\!P4}$. Moreover, by taking the limit as $s \rightarrow \infty$ for $s \in \mathcal{S}$ in \eqref{C1}, we have
\begin{align}
\setlength\abovedisplayskip{4pt}
\setlength\belowdisplayskip{4pt}
  &\!\!f(\overline{\mathbf{W}},\overline{\mathbf{V}}) \!\!\overset{(c)}{\leq}\!\! f(\overline{\mathbf{W}},\overline{\mathbf{V}})\!\!+\!\!\lim\limits _{s \in \mathcal{S}} \frac{1}{ \!\rho_{s}}\!\left(\left\|\mathbf{W}_{s}\right\|_{*}
  \!\!-\!\!\left\|\mathbf{W}_{s}\right\|_{2}\!\!+\!\!\left\|\mathbf{V}_{s}\right\|_{*}
  \!\!-\!\!\left\|\mathbf{V}_{s}\right\|_{2}\right)\notag\\
  &\leq f\left(\mathbf{W}^{\star},\mathbf{V}^{\star}\right),
\end{align}
where $(c)$ holds as the penalty factor $\rho_{s}$ and $\left\|\mathbf{W}_{s}\right\|_{*}
  -\left\|\mathbf{W}_{s}\right\|_{2}$, $\left\|\mathbf{V}_{s}\right\|_{*}-\left\|\mathbf{V}_{s}\right\|_{2}$
are nonnegative. Note that $\overline{\mathbf{W}},\overline{\mathbf{V}}$ are feasible for problem $\mathrm{SCA\!-\!P4}$ and its objective function $f(\overline{\mathbf{W}},\overline{\mathbf{V}})$ is no larger than $f\left(\mathbf{W}^{\star},\mathbf{V}^{\star}\right)$ obtained from optimal solutions $\mathbf{W}^{\star},\mathbf{V}^{\star}$ of problem $\mathrm{SCA\!-\!P4}$. Thus, $\overline{\mathbf{W}},\overline{\mathbf{V}}$ are also optimal to problem $\mathrm{SCA\!-\!P4}$. This completes the proof.
\end{appendices}

{\smaller[1]
	\bibliographystyle{IEEEtran}
	\bibliography{Reference}}

\begin{IEEEbiography}[{\includegraphics[width=1in,height=1.25in,clip,keepaspectratio]{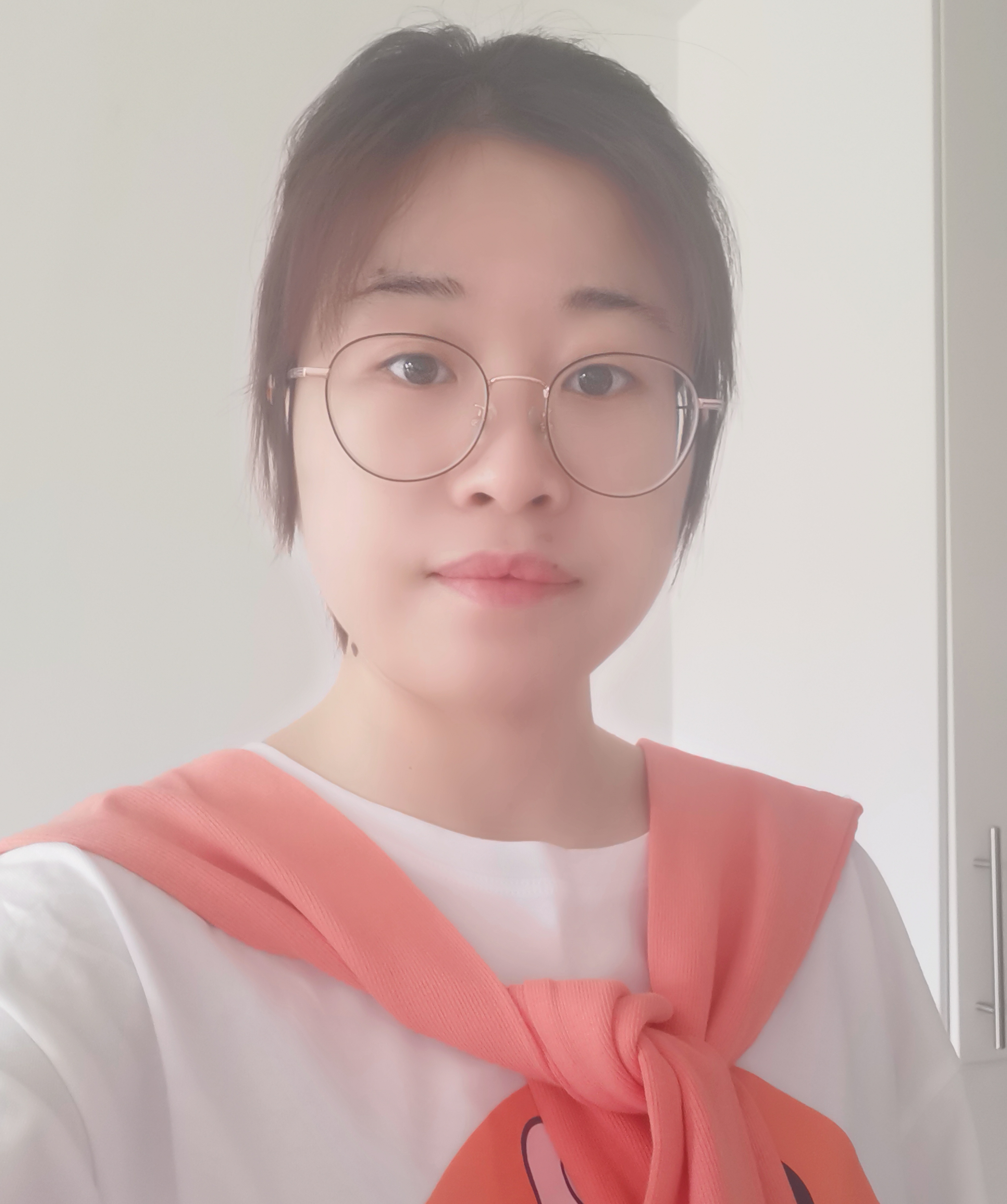}}]{Jing Cheng}
received the B.Eng. degree in communication engineering from Jiangsu University of Science and Technology, Zhenjiang, China, in 2017. She is currently pursuing the Ph.D. degree in information and communication engineering from the State Key Laboratory of Rail Traffic Control and Safety, Beijing Jiaotong University, Beijing, China. From April 2021 to April 2022, she was a visiting Ph.D. student at the Link\"oping University, Norrk\"oping, Sweden. Her main research interests include physical-layer resource allocation, ultra-reliable and low-latency communications (URLLC), and intelligent reflecting surface (IRS)-assisted communication.
\end{IEEEbiography}

\begin{IEEEbiography}[{\includegraphics[width=1in,height=1.25in,clip,keepaspectratio]{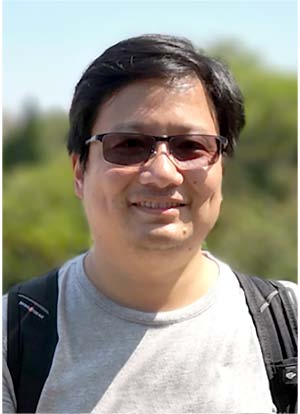}}]{Chao Shen}
(S'12-M'16) received the B.S. degree in communication engineering and the Ph.D. degree in signal and information processing from Beijing Jiaotong University (BJTU), Beijing, China, in 2003 and 2012, respectively. He was a Visiting Scholar with the University of Maryland at College Park, College Park, MD, USA, from 2014 to 2015, and the Chinese University of Hong Kong, Shenzhen, from 2017 to 2018. He has been an Associate Professor with the State Key Laboratory of Rail Traffic Control and Safety, BJTU since 2012. He has been working as a Senior Research Scientist with the Shenzhen Research Institute of Big Data, Shenzhen, China, since 2022.
His current research interests include large-scale network optimization, ultrareliable and low-latency communication (URLLC), and integrated sensing and communication (ISAC) for B5G/6G communications.
\end{IEEEbiography}

\begin{IEEEbiography}[{\includegraphics[width=1in,height=1.25in,clip,keepaspectratio]{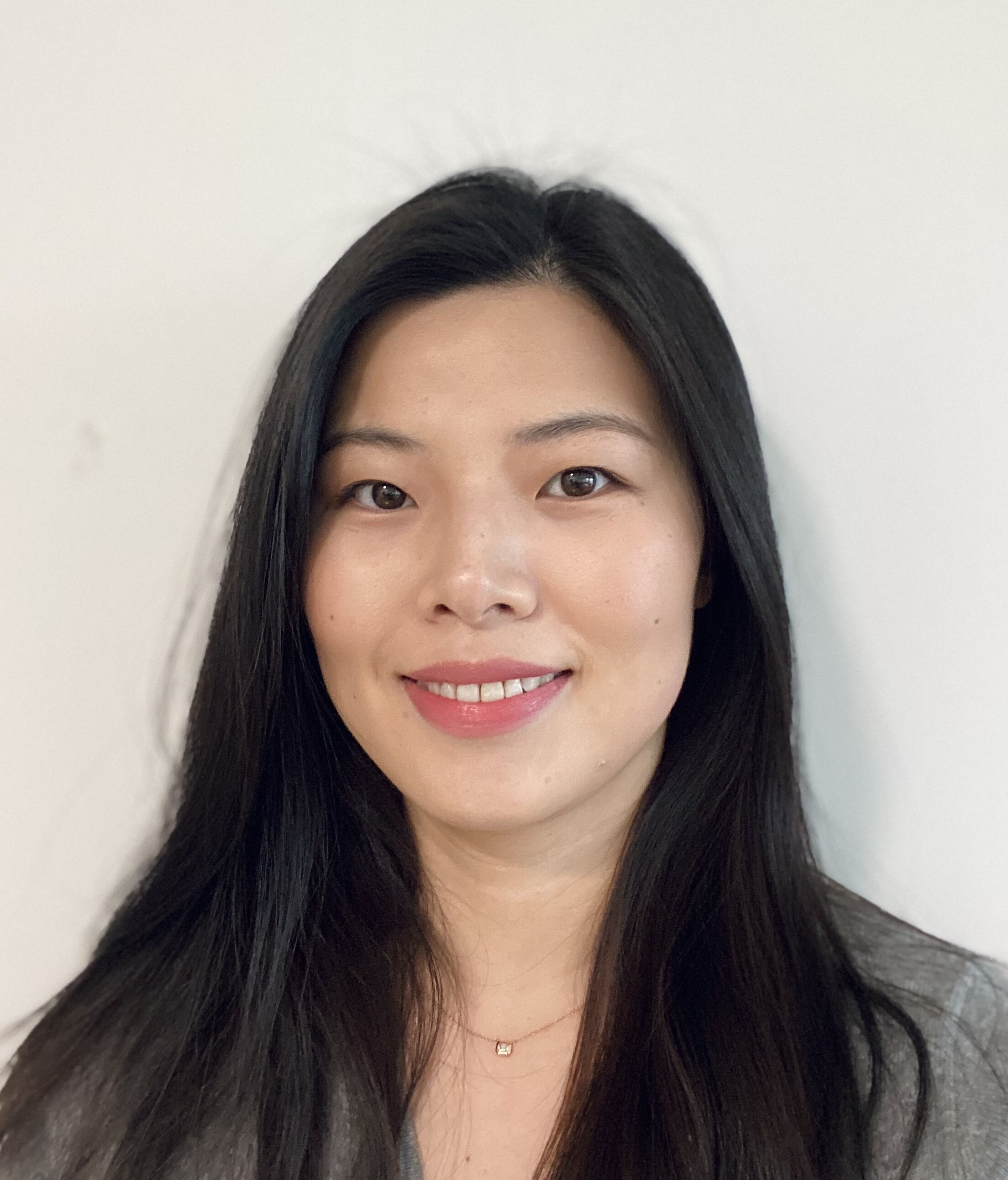}}]{Zheng Chen}(Member, IEEE) is an Assistant Professor with the Department of Electrical Engineering, Link\"oping University, Sweden. She received the B.Sc. degree from Huazhong University of Science and Technology (HUST), China, in 2011. Then she received the M.Sc. and Ph.D. degrees from CentraleSup\'elec, Universit\'e Paris-Saclay, France, in 2013 and 2016, respectively. From June to November 2015, she was a visiting scholar at Singapore University of Technology and Design (SUTD), Singapore. Since January 2017, she has been with Link\"{o}ping University, Sweden. Her main research interests include wireless communications, distributed intelligent systems, and network science.

She was the recipient of the 2020 IEEE Communications Society Young Author Best Paper Award. She was selected as an Exemplary Reviewer for IEEE Communications Letters in 2016, for IEEE Transactions on Wireless Communications in 2017, and for IEEE Transactions on Communications in 2019. She served as the workshop co-chair of the IEEE GLOBECOM Workshop on Wireless Communications for Distributed Intelligence in 2021 and 2022.
\end{IEEEbiography}

\begin{IEEEbiography}[{\includegraphics[width=1in,height=1.25in,clip,keepaspectratio]{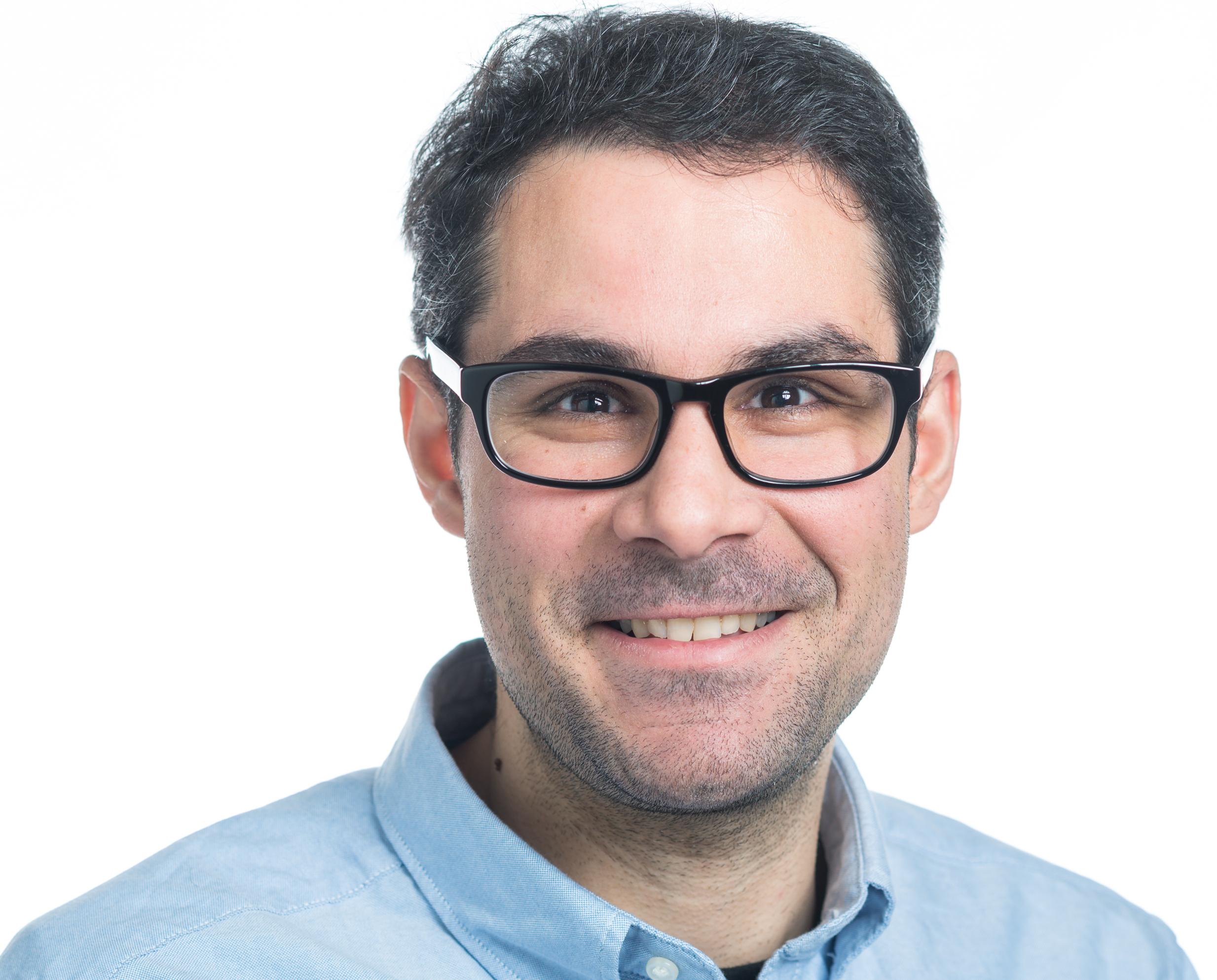}}]{Nikolaos Pappas}(Senior Member, IEEE) is an Associate Professor with the Department of Computer and Information Science, Link\"oping University, Sweden. He received the B.Sc. degree in computer science, the B.Sc. degree in mathematics, the M.Sc. degree in computer science, and the Ph.D. degree in computer science from the University of Crete, Greece, in 2005, 2012, 2007, and 2012, respectively. From 2005 to 2012, he was a Graduate Research Assistant with the Telecommunications and Networks Laboratory, Institute of Computer Science, Foundation for Research and Technology Hellas, and a Visiting Scholar with the Institute of Systems Research, University of Maryland at College Park, College Park, MD, USA. From 2012 to 2014, he was a Postdoctoral Researcher with the Department of Telecommunications, Supelec. His main research interests include the field of wireless communication networks with emphasis on the semantics-aware communications, energy harvesting networks, network-level cooperation, age of information, and stochastic geometry. From 2013 to 2018, he was an Editor of the IEEE Communications Letters. He was a guest editor for the IEEE Internet of Things Journal on ``Age of Information and Data Semantics for Sensing, Communication and Control Co-Design in IoT''. He is currently an Editor of the IEEE Transactions on Communications, the IEEE/KICS Journal of Communications and Networks, the IEEE Open Journal of Communications Society, and Expert editor for invited papers of the IEEE Communications Letters. He has served as a symposium co-chair for the IEEE International Conference on Communications 2022, and the IEEE Wireless Communications and Networking Conference 2022.
\end{IEEEbiography}

\ifCLASSOPTIONcaptionsoff
  \newpage
\fi

\end{document}